\documentclass[12pt]{article}
\usepackage[utf8]{inputenc}
\usepackage[margin=2cm]{geometry}
\usepackage{fullpage,enumitem,amssymb,amsmath,amsthm,xcolor,cancel,gensymb,hyperref,graphicx,bbm}
\usepackage{caption}
\usepackage{subfigure}
\usepackage{comment}

\usepackage{indentfirst}

\usepackage{natbib}
\usepackage{algorithm}
\usepackage{algpseudocode}

\usepackage{hyperref}
\usepackage{cleveref}
\usepackage{setspace}
\usepackage{titlesec}
\titlespacing\section{0pt}{6pt plus 4pt minus 2pt}{0pt plus 2pt minus 2pt}
\titlespacing\subsection{0pt}{6pt plus 4pt minus 2pt}{0pt plus 2pt minus 2pt}
\titlespacing\subsubsection{0pt}{6pt plus 4pt minus 2pt}{0pt plus 2pt minus 2pt}

\newtheorem{thm}{Theorem}
\newtheorem{lem}{Lemma}
\newtheorem{rmk}{Remark}
\newtheorem{cor}{Corollary}
\newtheorem{prop}{Proposition}
\newtheorem{defn}{Definition}
\newtheorem{assm}{Assumption}
\crefname{lem}{Lemma}{Lemmas}

\DeclareMathOperator*{\argmax}{arg\,max}
\DeclareMathOperator*{\argmin}{arg\,min}
\DeclareMathOperator*{\Var}{\text{Var}}

\DeclareMathOperator*{\bbE}{\mathbb{E}}
\DeclareMathOperator*{\bbR}{\mathbb{R}}

\DeclareMathOperator*{\bbP}{\mathbb{P}}
\DeclareMathOperator*{\calB}{\mathcal{B}}

\DeclareMathOperator*{\calE}{\mathcal{E}}

\DeclareMathOperator*{\calL}{\mathcal{L}}
\DeclareMathOperator*{\calN}{\mathcal{N}}

\DeclareMathOperator*{\calR}{\mathcal{R}}
\DeclareMathOperator*{\calS}{\mathcal{S}}

\DeclareMathOperator*{\tr}{\textup{tr}}
\DeclareMathOperator*{\rk}{\textup{rk}}
\DeclareMathOperator*{\vect}{\textup{vec}}

\DeclareMathOperator*{\br}{\mathbf{r}}

\providecommand{\keywords}[1]
{
  \small	
  \textbf{\textit{Keywords---}} #1
}

\title{Low-Rank Online Dynamic Assortment with Dual Contextual Information}

\author{
    Seong Jin Lee\thanks{Ph.D. student, Department of Statistics and Operations Research, University of North Carolina, Chapel Hill. Email: slee7@unc.edu} 
    \and
    Will Wei Sun\thanks{Associate Professor, Daniels School of Business, Purdue University. Email: sun244@purdue.edu} 
    \and
    Yufeng Liu\thanks{Professor, Department of Statistics, University of Michigan, Email: yufliu@umich.edu.}
}

\date{}

\usepackage{xr-hyper}

\begin{document}

\newpage

\maketitle

\onehalfspacing

\begin{abstract}
    As e-commerce expands, delivering real-time personalized recommendations from vast catalogs poses a critical challenge for retail platforms. Maximizing revenue requires careful consideration of both individual customer characteristics and available item features to continuously optimize assortments over time. 
    In this paper, we consider the dynamic assortment problem with dual contexts -- user and item features. 
    In high-dimensional scenarios, the quadratic growth of dimensions complicates computation and estimation.
    To tackle this challenge, we introduce a new low-rank dynamic assortment model to transform this problem into a manageable scale. Then we propose an efficient algorithm that estimates the intrinsic subspaces and utilizes the upper confidence bound approach to address the exploration-exploitation trade-off in online decision making. Theoretically, we establish a regret bound of $\tilde{O}((d_1+d_2)r\sqrt{T})$, where $d_1, d_2$ represent the dimensions of the user and item features respectively, $r$ is the rank of the parameter matrix, and $T$ denotes the time horizon. This bound represents a substantial improvement over prior literature, achieved by leveraging the low-rank structure. Extensive simulations and an application to the Expedia hotel recommendation dataset further demonstrate the advantages of our proposed method.
\end{abstract}
\keywords{Bandit algorithm, low-rankness, online decision making, regret analysis, reinforcement learning}

\newpage
\doublespacing
\setlength{\abovedisplayskip}{6pt minus 2pt}
\setlength{\belowdisplayskip}{6pt minus 2pt}
\setlength{\abovedisplayshortskip}{6pt minus 2pt}
\setlength{\belowdisplayshortskip}{6pt minus 2pt}


\baselineskip=23.5pt

\section{Introduction}

The assortment selection problem is a key challenge for retailers, who often face constraints on the number of items they can offer customers despite having a wide array of products. In the dynamic e-commerce environment, rich real-time information about customers and items drives the need for personalized algorithms adept at adapting to changing conditions. In online retail, the platform determines product recommendations by leveraging past assortment offerings and the users' choices on their assortment offerings. Figure \ref{fig:illust_assortment} illustrates the contextual dynamic assortment problem: At each time, a customer arrives with user features, the platform estimates user preferences on items based on historical data and current features, selects a subset from the item catalog to present to the current user and records user choices for improving future assortment selection.

\begin{figure}[!htb]
    \centering
    \captionsetup{width=.9\linewidth}
    \includegraphics[width = 0.7\textwidth]{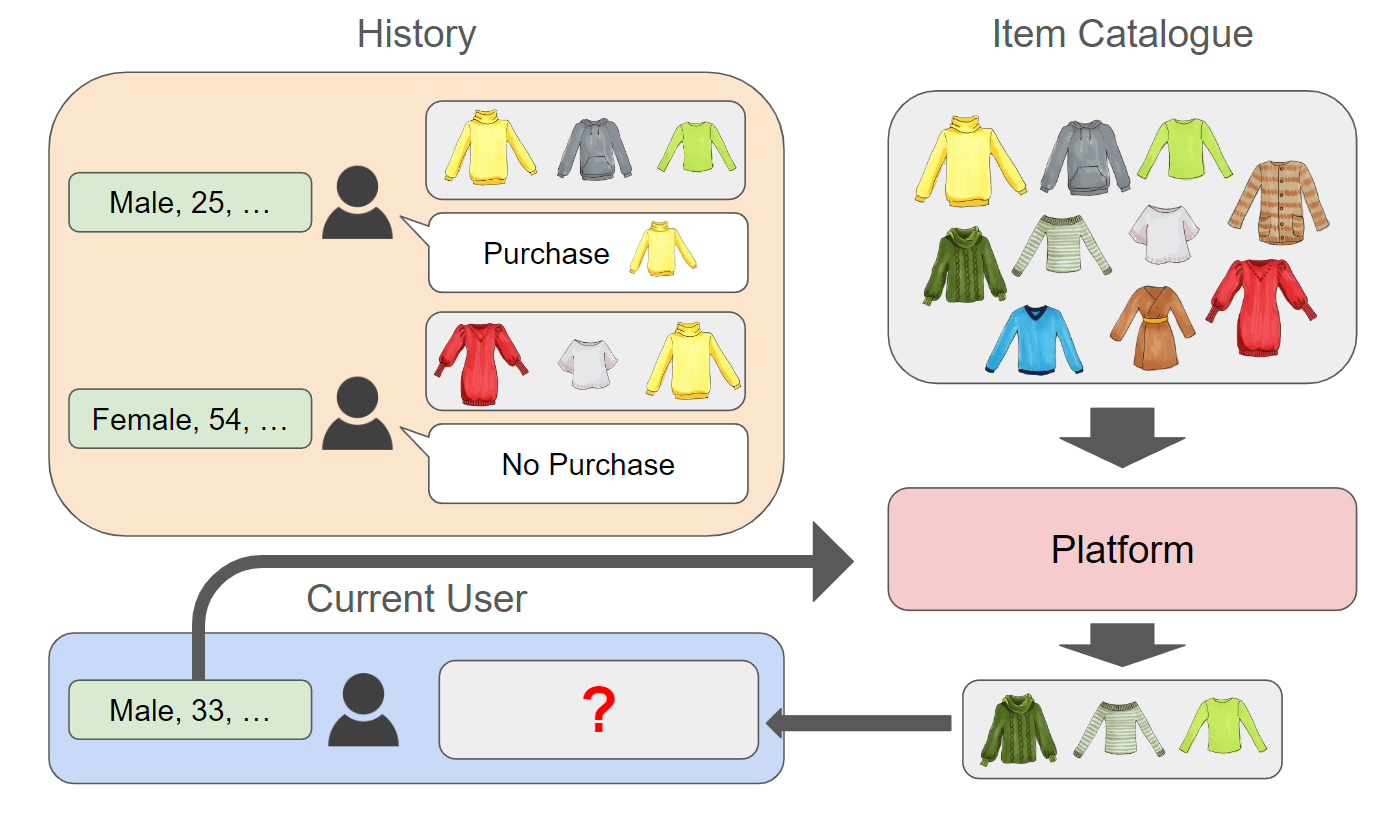}
    \caption{Illustration of the contextual dynamic assortment problem.}
    \label{fig:illust_assortment}
\end{figure}
\vspace{-1em}

The primary objective of the platform is to maximize the total revenue over a time horizon $T$. It can adopt a `greedy strategy' aimed at maximizing the revenue for each individual user by offering what is estimated as the best assortment at each time. However, due to the uncertainty of estimation, pure exploitation would lead to sub-optimal actions. To improve the accuracy of estimating these optimal assortments, the platform must delve into exploring user behaviors across various items. These two strategies, exploration for improving user preference estimation and exploitation for revenue maximization, often do not align in online decision making problems. Reinforcement learning, an area gaining increasing attention in statistics \citep{chakraborty2014dynamic, zhao2015new, shi2018high, shi2023multiagent, zhou2024estimating}, provides valuable insights for navigating the delicate balance of the exploration-exploitation trade-off.

In recent years, a plethora of studies have utilized bandit and reinforcement learning algorithms \citep{lattimore2020bandit} to address the dynamic assortment problem. Some research focuses on the non-contextual scenario \citep{chen2018note, agrawal2019mnl}, while others delve into the contextual setups \citep{chen2020dynamic, oh2021multinomial, goyal2022dynamic} under the multinomial logit (MNL) choice model \citep{mcfadden1974conditional} to characterize user behaviors. Advances in digital retail platform infrastructure have empowered these platforms with rich data on items and users, including demographics, search histories, and purchase records. By harnessing user features, the platform can tailor personalized assortments by looking back at the choice history of users with similar user contexts. Similarly, utilizing item features allows the platform to infer item preferences by analyzing the choice history of items sharing similar item contexts, even when the item has limited inclusion in assortment history. This motivates us to study the contextual dynamic assortment problem, specifically focusing on the incorporation of ``dual contextual information'' involving both items and users.

In contextual dynamic assortment, a common strategy for handling dual contexts is to stack the dual features as a joint feature vector and utilize its linear form \citep{sumida2023optimizing}. Let $q_t \in \bbR^{d_1}$ be the vector representation of the feature of user $t$ and $p_i \in \bbR^{d_2}$ be a vector representation of the feature of item $i$. By stacking the dual features as $x_{it} = (1, p_i^\top , q_t^\top)^\top \in \bbR^{d_1+d_2+1}$, the utility of item $i$ for user $t$ is represented as $v_{it} = x_{it}^\top \phi$ with vector parameter $\phi \in \bbR^{d_1+d_2+1}$; see Figure \ref{fig:StackUCBMNL_v}. 
\begin{figure}[h!]
    \centering
    \includegraphics[width = 0.7\textwidth]{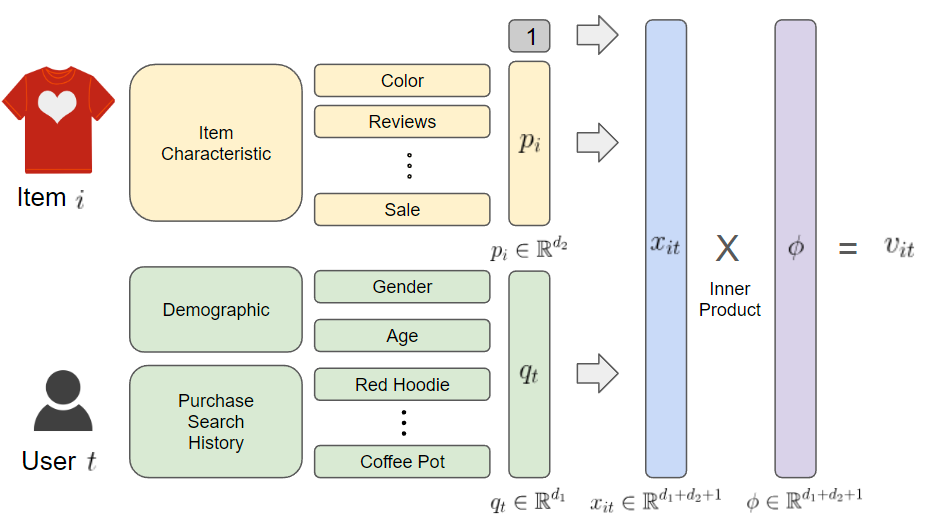}
    \caption{Stacked utility formulation of item and user features.}
    \vspace{-1em}
    \label{fig:StackUCBMNL_v}
\end{figure}
However, this approach fails to capture the interaction terms between the user feature and the item feature. Consider the effect of item price on the utility of an item to a user. The effect will be different across different users, some users might be sensitive to the price, while others may be not. Similarly, the red color of an item may have a positive effect on the utility if the user prefers red items, but will have a negative effect if the user dislikes red items. The additive nature of the stack formulation cannot capture this intrinsic interaction between the two sides of features.

Another approach involves utilizing the vectorization of the combined effects of the two features as a joint feature vector, followed by the application of existing contextual dynamic assortment selection methods \citep{agrawal2019mnl, oh2021multinomial}. By including the intercept and taking the outer product of the two extended features, it accounts for the intercept, main effects, and user-item interaction effects simultaneously; see Figure \ref{fig:VectUCBMNL_v}. However, in high-dimensional settings, the dimension of the outer product is $(d_1+1)\times (d_2+1)$, which scales with the product of the dimensions of the dual features, making the estimation of the effects computationally challenging.

\begin{figure}[h!]
    \centering
    \includegraphics[width = 1.0\textwidth]{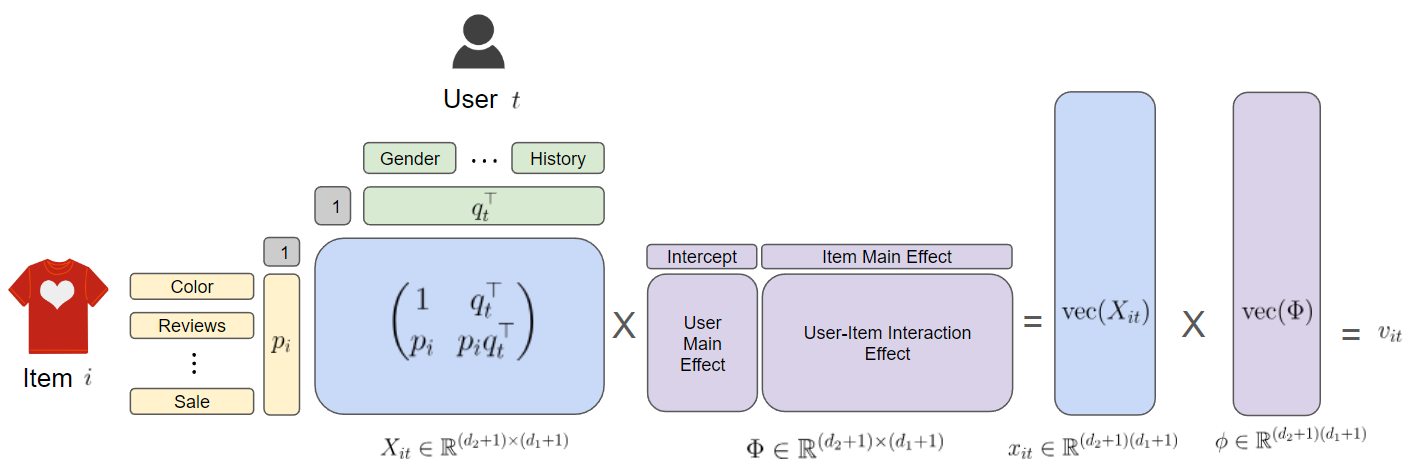}
    \caption{Vectorized utility formulation of interactions between user and item features.}
    \vspace{-1em}
    \label{fig:VectUCBMNL_v}
\end{figure}

In online retail, high-dimensional features can often be condensed into lower-dimensional latent factors. For instance, user traits like buying power can be inferred as a combination of income, education level, and past purchase behaviors. This buying power directly influences item utility, interacting with features such as price. Similarly, color preferences can be estimated from demographic data and previous red item purchases. Despite a user's extensive purchase history spanning numerous items, its impact on red item utility can be simplified into a scalar value. Likewise, an item's characteristics can be projected onto a reduced-dimensional latent space. Consequently, interaction effects between these dual contexts can be characterized by a smaller set of latent factors, as depicted in Figure \ref{fig:parametermatrix}. 
\vspace{-0.5em}
\begin{figure}[h!]
    \centering
    \includegraphics[width = 0.6\textwidth]{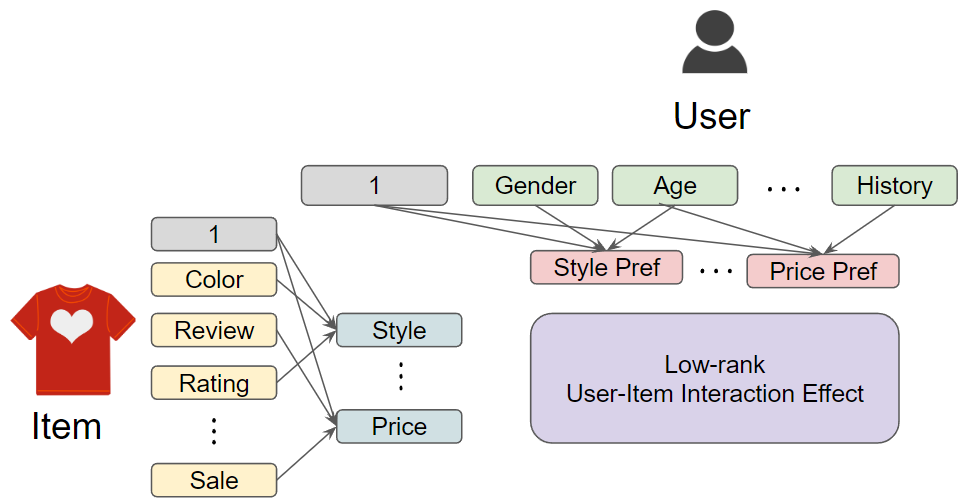}
    \caption{Illustration of low-rank user-item interaction effect.}
    \vspace{-1em}
    \label{fig:parametermatrix}
\end{figure}
By addressing the dual contextual problem within a low-rank structure, while retaining the impact of the features on the user preferences for items, we effectively reduce the dimension of the target parameter. This not only enhances computational efficiency but also reduces the cumulative estimation error across parameter components.

Inspired by this observation, in this paper we adopt the bilinear form $v = p^\top \Phi q$ with a low-rank matrix parameter $\Phi$ to represent utilities. This form offers flexibility, accommodating both stacking and vectorization approaches of incorporating dual features. 
The stack approach can be interpreted as the bilinear form with intercepts and a rank-2 parameter matrix while the vectorization approach can be interpreted as the scenario with full rank. Moreover, we can reduce the computational complexity by imposing a low-rank structure on $\Phi$, reducing the dimension of the parameter from $d_1d_2$ to $(d_1+d_2)r$, where $r(\ll \min\{d_1, d_2\})$ is the rank of the parameter matrix. Low-rank models have garnered growing interest in statistics \citep{jain2013low, xia2021statistical, xia2022inference, ma2024statistical} with diverse scientific and business applications \citep{bell2007lessons, koren2009matrix, udell2019big, zhen2024nonnegative, zhou2024stochastic}. Following this trend, our paper considers a new Low-rank Online Assortment with Dual-contexts (LOAD) problem.

To solve the LOAD problem, we propose an Explore-Low-rank-Subspace-then-Apply-UCB (ELSA-UCB) algorithm which consists of two main stages: subspace exploration and the application of Upper Confidence Bound (UCB). The major computational bottleneck arises from estimating the entire preference parameter matrix.
To overcome this challenge, we project the features and parameters to a low-dimensional subspace. In the first stage, we leverage the inherent low-rank structure to estimate the low-dimensional subspace of the matrix parameter. To accomplish this, we explore the space by employing random assortments and then solve the rank-constrained likelihood maximization problem to estimate the subspace. Given the non-convex nature of the rank-constrained optimization problem, we adopt the alternating gradient descent algorithm based on the Burer-Monteiro formulation \citep{burer2003nonlinear, zheng2016convergence, chi2019nonconvex} to estimate the parameter matrix. Subsequently, we perform singular value decomposition (SVD) on the estimated matrix to derive the intrinsic subspace. Utilizing the estimated subspace, we rotate the features and parameters in accordance with the subspace, then truncate negligible terms to effectively reduce the problem's dimensionality. In the second stage, we employ the UCB-based strategy on the reduced dimension based on careful construction of the confidence bounds in this space. UCB uses an ``optimism in the face of uncertainty'' idea that addresses the exploration-exploitation trade-off by selecting the best action based on the ``optimistic reward'', the upper confidence bound of the estimated reward \citep{auer2002using, lattimore2020bandit}. Contrary to the unbiased model in the original parameter space, the preference model within the reduced space is inherently biased with respect to the estimated subspaces. To handle this, we introduce a novel tool that provides a confidence bound for the expected reward on the reduced space, by correcting this bias. Implementing the UCB-based policy in this reduced space reduces the parameter dimension from $d_1 \times d_2$ to $r(d_1 + d_2) - r^2$. This reduction markedly improves computational efficiency, reduces estimation error, and enhances the algorithm performance. In practice, the rank of the underlying parameter matrix is often unknown. To address this issue, we propose ELSA-GIC, incorporating the rank selection procedure using Generalized Information Criterion (GIC) \citep{konishi1996generalised, fan2013tuning} into ELSA-UCB. We also prove rank selection consistency of ELSA-GIC in the LOAD problem.

To assess the theoretical performance of our proposed algorithm, we study the cumulative regrets of our policy. The regret is defined as the difference between the optimal reward from the oracle possessing complete knowledge of environmental parameters and the reward obtained from a given policy. In theory, we establish an $\tilde{O}((d_1+d_2)r\sqrt{T})$ regret upper bound for our proposed ELSA-UCB policy (Theorem \ref{thm:regretbound}). Notably, in low-rank setups with $r\ll \min\{d_1, d_2\}$, our regret bound is much improved from $\tilde{O}(d_1d_2\sqrt{T})$ regret bounds achieved by existing methods \citep{chen2020dynamic, oh2021multinomial} that vectorize the outer product of the dual features.

The regret analysis for our policy presents several challenges. Firstly, in the dual-contextual environment, the distribution of joint features depends on the assortment selection, violating the i.i.d. feature assumption in existing dynamic assortment selection literature \citep{chen2020dynamic, oh2021multinomial}. To accommodate this, we impose minimal assumptions on the item features (Assumption \ref{assm:item_feature}) to ensure the identifiability of parameters. Secondly, the non-convex nature of the low-rank optimization problem precludes us from leveraging classical convex optimization theories. While \citet{jun2019bilinear, kang2022efficient} studied generalized linear bandits with low-rank structure, their analysis tools are not applicable to our LOAD problem due to the unique combinatorial structure in the dynamic assortment selection. Thus, we establish novel estimation bounds specifically tailored for the rank-constrained optimization problem on assortment scenarios.

Finally, we present numerical results to demonstrate the effectiveness of our algorithm. Initially, we evaluate the ELSA-GIC policy's performance on synthetic data, exploring various scenarios including changing ranks, feature dimensions, number of items, and assortment capacity. We compare cumulative regret over different time horizons with two variants of the UCB-MNL policy \citep{oh2021multinomial}: Stacked UCB-MNL and Vectorized UCB-MNL. Our analysis consistently shows that ELSA-GIC outperforms both methods across all scenarios, with the performance gap widening in scenarios with lower ranks and higher dimensions. We also show the robust performance of ELSA-GIC under model mis-specification. Additionally, we conduct real data analysis on the Expedia dataset to optimize hotel recommendations. Our findings demonstrate that the parameter matrix exhibits a low-dimensional structure, and our method effectively leverages this to achieve improved performance, highlighting its practical usefulness.

\subsection{Related Work}

Our work is closely related to recent studies on contextual dynamic assortment selection and low-rank bandits. Additional relevant literature is provided in Section S.1 of the supplement.

\begin{itemize}
\item \textbf{Contextual Dynamic Assortment Selection}

Inspired by the foundational work by \citet{caro2007dynamic}, numerous studies have explored the dynamic assortment problem under the MNL bandit model. Various policies, including Explore-Then-Commit \citep{rusmevichientong2010dynamic, saure2013optimal}, Thompson Sampling \citep{agrawal2017thompson} and UCB \citep{agrawal2019mnl} have been proposed to address the balance between exploration and exploitation in the MNL bandit model. Other studies consider various scenarios such as robustness against outliers \citep{chen2023robust}, multi-stage choices \citep{xu2023assortment} or extension to different choice models \citep{aouad2023exponomial}. Contextual dynamic assortment considers various facets, including item-specific features \citep{wang2019online}, time-varying item features \citep{chen2020dynamic}, user types \citep{kallus2020dynamic}, item features \citep{shao2022sparse}, or features associated with user-item pairs \citep{oh2021multinomial, goyal2022dynamic}. These features are utilized in modeling utilities for the MNL bandit model. However, while existing algorithms address single-context scenarios, focusing solely on users or items, they lack the capability to estimate interactions between dual contexts. Although algorithms designed for joint features on user-item pairs could be applied to the dual contextual setup by interpreting the vectorization of the outer product of dual features, this approach inevitably encounters large computational burdens and fails to capture the low-rank structure of the parameter matrix in our LOAD problem. Our proposed method aims to fill this gap.

\item \textbf{Low-Rank Bandits}

The dynamic environment involving high-dimensional covariates has also been a vibrant area of research. Based on a rich background of statistical tools for high dimensional problems \citep{wainwright2019high}, numerous studies addressed the `curse of dimensionality' in high-dimensional bandits. These studies often assume sparsity in the parameter vector and employ the LASSO-based approach \citep{abbasi2012online, kim2019doubly} or presume the low-rank structure within the parameter matrix \citep{jun2019bilinear, kang2022efficient, cai2023doubly, zhou2024stochastic}. While certain studies have utilized the LASSO-based approach to address the sparse dynamic assortment problem \citep{wang2019online, shao2022sparse}, there has been comparatively less emphasis on leveraging the low-rank structure in the dynamic assortment. Recently, \citet{cai2023doubly} approached the dynamic assortment problem by modeling rewards as a bilinear form of action and context. However, they did not consider any user choice model and overlooked the inherent structure between assortment sets. Because of this, their algorithms and analysis tools are not applicable to our problem. To the best of our knowledge, our approach is the first one to incorporate the low-rank structure in the dynamic assortment selection with dual contexts.

\end{itemize}

\subsection{Notations}

Let $[n]$ denote the set $\{1,2, \ldots, n\}$. The Frobenius norm of a matrix is denoted by $\|\cdot\|_F$, and the matrix-induced norm  $\|v\|_W = \sqrt{v^\top W v}$ is denoted by $\|\cdot\|_W$. The ball induced by the Frobenius norm, centered at $\Theta_0$ with radius $r$ is denoted by $\calB(\Theta_0, r) := \{\Theta: \|\Theta - \Theta_0\|_F \le r\}$. For a vector $v$, $(v)_i$ represents its $i$-th component. We write $\Sigma_1 \succ \Sigma_2$ if $\Sigma_1 - \Sigma_2$ is positive definite. Denote the Kronecker product of two matrices by $\otimes$. We use the notation $\tilde{O}(a_n)$ to denote that the quantity is bounded by $a_n$ up to logarithmic factors.

\section{Problem Formulation}


In this section, we introduce the mathematical formulation of the LOAD problem. At time $t$, a user with feature vector $q_t \in \bbR^{d_1}$ arrives at the platform. Assume there are $N$ products, each with item feature vector $p_i \in \bbR^{d_2}$ for $i\in [N] = \{1, \ldots , N\}$.  The platform offers an assortment $S_t$ of size at most $K$ from the catalog $[N]$. In other words, $S_t\subset [N]$ and $|S_t| \le K$. When the user is provided with an assortment, the user either chooses one of the items in the assortment or chooses not to purchase at all. To quantify this user choice, let $y_{it} \in \{0,1\}$ represent whether user $t$ chooses item $i$, for $i\in [N]$ and $t\ge 1$. In addition, denote $y_{0t}=1$ as the indicator that user $t$ does not choose any item at all. Let $i_t\in S_t$ be the item that user $t$ chooses ($0$ if no purchase). Note that for each user $t$, $y_{i_t,t} = 1$ and the remaining is $y_{jt} = 0$ for $j \in S_t \cup \{0\} \backslash \{i_t\}$, i.e., $\sum_{i \in S_t\cup\{0\}} y_{it} = 1$. Denote $\textbf{y}_t = (y_{1t}, \ldots , y_{Nt})^\top \in \{0,1\}^N$ as the choice vector of user $t$.

In assortment selection, a widely adopted choice model is the multinomial logit (MNL) model \citep{mcfadden1974conditional, rusmevichientong2010dynamic, agrawal2019mnl}. Suppose the utility of item $i$ to user $t$ is given as a random variable $X_{it} = v_{it} + \epsilon_{it}$, where $v_{it}$ is the expected utility and $\epsilon_{it}$ is a mean-zero error term. 
Assume $v_{0t} = 0$ for identifiability. The MNL model assumes that 
the probability of user $t$ choosing item $i$ is given as
$$
p_t(i) := \bbP(y_{it} = 1) = \frac{\exp(v_{it})}{1+\sum_{j\in S_t} \exp(v_{jt})}, \quad  i \in S_t,
$$
and the probability of no purchase is given as $1/(1+\sum_{j\in S_t}\exp(v_{jt}))$. In this paper, we model the utility $v_{it}$ as a low-rank bilinear form of the dual contexts with $$v_{it} = p_i^\top \Phi^* q_t,$$ where $\Phi^*\in \bbR^{d_2 \times d_1}$ is a rank-$r$ matrix, illustrated in Figure \ref{fig:parametermatrix}, with singular values $\sigma_1 \ge \cdots \ge\sigma_r > 0$. 

After the user's purchasing decision under the MNL model, the platform earns revenue from the selected item (or zero if no purchase is made). Let $r_{i}$ denote the revenue that the platform achieves when item $i$ is sold. 
Given an assortment $S$, the revenue that the platform achieves from user $t$ is $\sum_{i\in S} r_i y_{it}$, and the expected revenue of the platform for user $t$ is
$$
R_t(S, \Phi^*) = \bbE \left[\sum_{i\in S} r_{i} \cdot y_{it}\right] =  \sum_{i\in S} r_{i} p_t(i |S, \Phi^*).
$$

The platform aims to choose the personalized assortment for each user $t$ at time $t$ to maximize the cumulative expected revenue gained from all users. Note that the possible assortment set with constraint $K$, i.e., $\calS = \{S\subset [N]: |S|\le K\}$ is finite. Therefore, when $\Phi^*$ is known in hindsight, there exists an optimal assortment $S_t^*$ for user $t$ that maximizes the expected revenue since we can calculate the expected utility for every $S\in \calS$. Denote the optimal assortment at time $t$ as
$$
S_t^* = \argmax_{S \in \calS} R_t(S, \Phi^*).
$$
However, in practice, the true parameter $\Phi^*$ is not known to the platform. Instead, the platform chooses the assortment for each user $t$ based on some policy. To evaluate the performance of a policy, we define regret as the difference in expected revenue between the optimal assortment and the chosen assortment from the given policy. In other words, if we choose assortments $S_t$ for each user $t$, the total cumulative regret over time horizon $T$ is defined as
$$
{\calR}_T = \bbE \left[ \sum_{t=1}^T \left( R_t(S_t^*, \Phi^*) - R_t(S_t, \Phi^*) \right) \right].
$$

Note that maximizing the cumulative revenue is equivalent to minimizing the cumulative regret. At time $t$, given the current user context $q_t$ and historical data, including past assortments $S_1, S_2, \ldots, S_{t-1}$ presented to users $1, 2, \ldots, t-1$ with respective features $q_1, q_2, \ldots, q_{t-1}$, and the corresponding feedback $\textbf{y}_1, \textbf{y}_2, \ldots, \textbf{y}_{t-1}$, the platform needs to determine the assortment $S_t$ for the user $t$ so that the cumulative regret could be small.

\section{Algorithm}

In this section, we propose our Explore-Low-rank-Subspace-then-Apply-UCB (ELSA-UCB) method to address the LOAD problem defined in the previous section. The algorithm is mainly composed of two stages. The first stage is the exploration stage with random assortments which aims to estimate the low-dimensional subspace of the matrix parameter $\Phi^*$. Using the estimated subspace, we rotate the parameter and feature space and truncate the negligible dimensions. In the second stage, we use a UCB-based approach on the estimated subspace to select the assortment with the highest possible reward.

\subsection{Estimation of Subspace}

The first stage of the algorithm aims to acquire estimations of low-rank subspaces. Let $\Phi^* = U^* D^* (V^*)^\top$ be the singular value decomposition of $\Phi^*$, where $U^* \in \bbR^{d_2\times d_2}, V^* \in \bbR^{d_1\times d_1}$ are square orthonormal matrices and $D^*\in \bbR^{d_2\times d_1}$ is a diagonal matrix with singular values as its diagonal components. Let $u^*_{(j)}$ be the $j$-th column of $U^*$ and $v^*_{(j)}$ be the $j$-th column of $V^*$. Consider dual contexts $p\in \bbR^{d_2}$ and $q \in \bbR^{d_1}$. The bilinear form of $p$ and $q$ with respect to $\Phi^*$ can be rewritten as
$$
p^\top \Phi^* q = ((U^*)^\top p)^\top D^* ((V^*)^\top q) = \sum_{j=1}^r d_j p'_j q'_j,
$$
where $p'_j = (u^*_{(j)})^\top p$ and $q'_j = (v^*_{(j)})^\top q$ are rotations of $p$ and $q$ each with respect to orthonormal basis $U^*$ and $V^*$. Thus, the first $r$ columns of $U$ and $V$ each capture the $r$ latent features originating from dual contexts, which affect user preferences toward items. Consequently, by estimating the subspace of the matrix $\Phi^*$ via estimating $U^*$ and $V^*$, the platform effectively extracts the low-dimensional latent features, maintaining the representational capability of the original features.

For estimation of the subspace, we start with random assortments offered to the first $T_0$ users to achieve an initial estimate of $\Phi^*$. Consider the rank-constrained likelihood maximization problem:
\begin{equation} \label{eqn:rank_constrained_MLE}
\min_{\Phi \in \bbR^{d_2\times d_1}, \ \text{rank}(\Phi)\le r} {\calL}_n(\Phi),
\end{equation}
where $\calL_n(\Phi)$ is the negative log-likelihood of the MNL choice model with $n$ samples, expressed as 
$$
{\calL}_n(\Phi) = -\frac{1}{n}\sum_{t=1}^n \left[ \sum_{i\in S_t} y_{it} p_i^\top \Phi q_t - \log \left(1+ \sum_{j\in S_t} \exp(p_j^\top \Phi q_t) \right) \right].
$$

To solve this rank-constrained optimization, we use the Burer-Monteiro formulation \citep{burer2003nonlinear} to optimize over $U \in \bbR^{d_2\times r}$ and $V \in \bbR^{d_1\times r}$, where $\Phi = UV^\top$. We add a regularization term to enforce identifiability \citep{zheng2016convergence}. This leads to our final optimization problem:
\begin{equation} \label{eqn:Monteiro_Burer_MLE}
\min_{U \in \bbR^{d_2\times r}, \ V\in \bbR^{d_1\times r}} \left( \tilde{\calL}_n(U,V) :=  {\calL}_n(UV^\top) + \frac{1}{8} \|U^\top U - V^\top V\|_F^2 \right).
\end{equation}

This non-convex optimization can be solved effectively via an alternating factored gradient descent (FGD). FGD simply applies gradient descent on each factor,
\begin{align*}
    U \leftarrow U - \eta \cdot \nabla_U \tilde{\calL}_n(U,V), \quad V \leftarrow V - \eta \cdot \nabla_V \tilde{\calL}_n(U,V),
\end{align*}
where $\eta$ is the learning rate, $\nabla_U, \nabla_V$ are gradients with respect to each component when other parameters are fixed. Applying these gradient descent steps alternatively, the estimated parameters are known to converge to the true parameters given conditions on the initial estimate \citep{zheng2016convergence, zhang2023generalized}. Details of FGD are given in Algorithm \ref{alg:fgd}.

\begin{algorithm}
 \caption{FGD - Factored Gradient Descent}
 \label{alg:fgd}
 \begin{algorithmic}[1]
 \Require Gradients of the loss function $\nabla_U\tilde\calL$ and $\nabla_V\tilde\calL$, initial estimate $\Phi_0$ and step size $\eta$.
 \Ensure SVD $\Phi_0 = U_0 D_0 V_0^\top$ and initialize $U^{(0)} = U_0D_0^{1/2}$ and $V^{(0)} = V_0D_0^{1/2}$.
 \For{$k=1, \ldots$,}
 \State Update $U^{(k+1)} = U^{(k)} - \eta \nabla_U {\tilde\calL}|_{U = U^{(k)}, V = V^{(k)}}$.
 \State Update $V^{(k+1)} = V^{(k)} - \eta \nabla_V {\tilde\calL}|_{U = U^{(k+1)}, V = V^{(k)}}$.
 \State Repeat until objective function converges.
 \EndFor
 \end{algorithmic}
\Return $\hat\Phi = U^{(k)} (V^{(k)})^\top$.
\end{algorithm}

For the initialization of FGD, we use the unconstrained maximum likelihood estimator (MLE) of $\Phi^*$. As the likelihood function is convex, we can apply a convex optimization algorithm such as the gradient-based algorithm or L-BFGS \citep{liu1989limited} to calculate the initialization.

With the estimate $\hat\Phi$ from Algorithm \ref{alg:fgd}, we obtain the estimated subspaces $\hat{U}$ and $\hat{V}$ via SVD, $\hat\Phi = \hat{U} \hat{D} \hat{V}^\top$. Next, we rotate the original subspace with respect to the estimated subspace. We partition $U^*, V^*, \hat{U}, \hat{V}$ into $(U^*_1, U_2^*), (V_1^*, V_2^*), (\hat{U}_1, \hat{U}_2), (\hat{V}_1, \hat{V}_2)$, where $(\cdot)_1$ denotes the first $r$ columns and $(\cdot)_2$ denotes the remaining columns of a matrix. Define
\begin{align*}
    \Theta^* &:= \hat{U}^\top \Phi^* \hat{V}
    = (\hat{U}_1, \hat{U}_2)^\top U_1^* D_{11} (V_1^*)^\top (\hat{V}_1, \hat{V}_2)  \\
    &= \begin{pmatrix}
        (\hat{U}_1^\top U_1^*) D_{11}^* (\hat{V}_1^\top V_1^*)^\top & (\hat{U}_1^\top U_1^*) D_{11}^* (\hat{V}_2^\top V_1^*)^\top \\
        (\hat{U}_2^\top U_1^*) D_{11}^* (\hat{V}_1^\top V_1^*)^\top  & (\hat{U}_2^\top U_1^*) D_{11}^* (\hat{V}_2^\top V_1^*)^\top 
    \end{pmatrix}.
\end{align*}

We adopt the ``subtraction method'' prevalent in the low-rank matrix bandit \citep{kveton2017stochastic, jun2019bilinear, kang2022efficient}. Note that $\hat{U}_2^\top U_1^*$ and $\hat{V}_2^\top V_1^*$ are expected to be small, assuming that $\hat{U}, \hat{V}$ are good estimations of $U^*, V^*$. Thus its product $(\hat{U}_2^\top U_1^*) (\hat{V}_2^\top V_1^*)^\top$ is expected to be negligible. To utilize this structure, we introduce the notation of ``truncation-vectorization'' of matrices. Consider $\Theta \in \bbR^{d_2\times d_1}$ with block-sub-matrices $\Theta_{11} \in \bbR^{r\times r}$, $\Theta_{12} \in \bbR^{r\times (d_1-r)}$, $\Theta_{21} \in \bbR^{(d_2-r)\times r}$, $\Theta_{22} \in \bbR^{(d_2-r)\times (d_1-r)}$. Define ``truncation-vectorization'' of $\Theta$ as
\begin{equation}\label{eqn:rtv}
\theta_{rtv} := (\vect(\Theta_{11})^\top , \vect(\Theta_{12})^\top , \vect(\Theta_{21})^\top)^\top \in {\bbR}^{(d_1+d_2)r-r^2},
\end{equation}
where ``$rtv$'' is an abbreviation of ``rotation-truncation-vectorization'' to represent that the original matrix $\Phi^*$ has been rotated with respect to $\hat{U}$ and $\hat{V}$, then truncated and vectorized.

Note that $\Theta^*_{22} = (\hat{U}_2^\top U_1^*) D_{11}^* (\hat{V}_2^\top V_1^*)^\top$ is expected to be negligible, so we only need to focus on the estimation of $\theta^*_{rtv}$, the truncation-vectorization of $\Theta^*$. The computational advantage of this is that the dimension of the parameter is now reduced from $d_2 d_1$ of $\Theta^*$ to $r(d_1+d_2)-r^2$ of $\theta^*_{rtv}$. In the rotated space, the user utility of an item is expressed as
$$
v_{it} = p_i^\top \Phi^* q_t = (\hat{U}^\top p_i)^\top \Theta^* (\hat{V}^\top q_t) = \langle (\hat{U}^\top p_i) (\hat{V}^\top q_t)^\top, \Theta^* \rangle.
$$

Denote $X_{it} = (\hat{U}^\top p_i) (\hat{V}^\top q_t)^\top$ and define the truncation-vectorization of $X_{it}$ as $x_{it,rtv}$. Then
$$
v_{it} = \langle X_{it}, \Theta^* \rangle = \langle x_{it,rtv}, \theta_{rtv} \rangle + \langle (X_{it})_{22}, \Theta^*_{22} \rangle \approx \langle x_{it,rtv}, \theta_{rtv} \rangle,
$$
as $\Theta^*_{22} \approx 0$. Hence we can represent our original dual features as $x_{it,rtv}$ and focus on the estimation of low-dimensional parameter $\theta_{rtv}$ for our dynamic assortment selection policy.

When the rank $r$ of the parameter matrix is unknown, we can estimate the rank using Generalized Information Criterion (GIC). GIC \citep{konishi1996generalised, fan2013tuning, morimoto2024information, park2024low} is a generalization of information criteria including Bayesian Information Criterion (BIC) and Akaike Information Criterion (AIC) \citep{akaike1987factor}, and is widely used in the high-dimensional estimation problems to select hyper parameters to avoid over-fitting. 

Adapting to our setting, we utilize the GIC in the following form:
\begin{equation} \label{eqn:GIC}
GIC(r) := {\calL}_n(\hat\Phi_r) + a_n \cdot (d_1 + d_2 - r) \cdot r,
\end{equation}
which penalizes the number of free parameters under the low-rank assumption. We show that the selection of $r$ is theoretically consistent with proper selection of $a_n$ (Proposition 2 in the Appendix) and demonstrate the accuracy of rank selection with $a_n=\log(n)/n$ in our numerical results (Figure \ref{fig:rankest}). We provide details of incorporating GIC to ELSA-UCB in Appendix Section S.5.2.

\subsection{UCB-based Stage}

The second stage of the algorithm employs the UCB approach on $\theta^*_{rtv}$. UCB algorithm, also known as ``optimism in the face of uncertainty'' \citep{lattimore2020bandit}, operates based on the principle that the decision-maker should select the assortment with the highest ``optimistic reward'', or the upper confidence bound of the estimated reward. Even if an assortment has a low estimated expected reward, when it has a high upper confidence bound due to large uncertainty, the algorithm can still favor that assortment. This approach inherently strikes a balance between exploiting the estimated reward and exploring assortments with higher uncertainty. As the algorithm is encouraged to further explore assortments with greater uncertainty, the confidence bounds gradually tighten, converging towards the true expected reward.

To adopt the UCB approach, we formulate the optimistic utility of item $i$ to user $t$ as
\begin{equation} \label{eqn:def_zit}
z_{it} = x_{it,rtv}^\top \widehat{\theta}_{t,rtv} + \beta_{it}.
\end{equation}
The first term $x_{it,rtv}^\top \widehat{\theta}_{t,rtv}$ is the estimate for utility of the item $i$ for user $t$, and the latter term $ \beta_{it}$ quantifies the uncertainty of the utility. For $\beta_{it}$, it can be calculated as $\beta_{it} = \alpha_t \cdot \|x_{it,rtv}\|_{W_t^{-1}}$, where the selection of $\alpha_t$ is based on the confidence bound of $\widehat{\theta}_{t,rtv}$ established in Lemma \ref{lem:thetabound} of Section \ref{sect:theory} and $W_t = \sum_{s=1}^t \sum_{i\in S_s} x_{is,rtv}x_{is,rtv}^\top$. Then we choose the optimal assortment that maximizes the revenue under the optimistic utilities. In other words, we offer the assortment $S_t = \argmax_{S \in \calS} \tilde{R}_t(S)$, where 
$$
\tilde{R}_t(S) := \sum_{i\in S} r_i \cdot \frac{\exp(z_{it})}{1 + \sum_{j\in S} \exp(z_{it})}.
$$
After obtaining the optimistic utilities, this problem can be solved via the StaticMNL policy introduced for the static assortment problem \citep{rusmevichientong2010dynamic}.

During the construction of confidence bounds for the optimistic reward, the algorithm simultaneously updates the estimate $\widehat{\theta}_{t,rtv}$. Distinct from the rank-constrained MLE in the first stage, in the UCB-based stage, we solve a new estimator for the updated parameter $\theta^*_{rtv}$. Since the parameter dimensionality has reduced from $d_1\times d_2$ in the first stage to $(d_1+d_2)r-r^2$ in this UCB-stage, we can calculate the MLE directly on the reduced space to obtain $\widehat{\theta}_{n,rtv} = \argmin \calL_{n, rtv}(\theta_{rtv})$, where $\calL_{n,rtv}$ is the negative log-likelihood in the reduced space with
\begin{align}
\label{eqn:loss_rtv}
{\calL}_{n,rtv}(\theta_{rtv}) = \frac{1}{n} \sum_{t=1}^n \left[ \sum_{i\in S_t} y_{it} x_{it,rtv}^\top \theta_{rtv} - \log\left( 1+ \sum_{j\in S_t} \exp(x_{it,rtv}^\top \theta_{rtv}) \right) \right].
\end{align}
It is crucial to note that this optimization problem is simpler than the rank-constrained MLE in (\ref{eqn:rank_constrained_MLE}) or (\ref{eqn:Monteiro_Burer_MLE}), since this optimization is unconstrained and many existing convex optimization techniques can be employed to solve (\ref{eqn:loss_rtv}).

\begin{algorithm}
 \caption{ELSA-UCB - Explore Low-dimensional Subspace then Apply UCB}
 \label{alg:ELSA-UCB}
 \begin{algorithmic}[1]
 \Require Assortment capacity $K$, rank of parameter matrix $r$, learning rate $\eta$, initial parameter estimate $\Phi_0$ and exploration length $T_0$.
 \State Observe item feature vectors $p_i$, $i\in [N]$.
 \For{$t = 1, \ldots, T_0$,} 
 \State Observe the current user feature vector $q_t$.
 \State Randomly select a size $K$ assortment $S_t \in \calS$ and observe user choice $\mathbf{y}_t$.
 \EndFor 
 \State Estimate the low-rank matrix $\hat\Phi = \argmax_{rk(\Phi)=r} {\calL}_n(\Phi)$ using Algorithm \ref{alg:fgd}.
 \State Estimate the subspace using SVD $\hat\Phi = \hat{U}\hat{D}\hat{V}^\top$.
 \State Initialize $\widehat{\theta}_{t,rtv} \in \bbR^k$ as the truncated vectorization of $\hat{D}$ as in (\ref{eqn:rtv}).
 \State Rotate the item features $p_i' = \hat{U}^\top p_i$ for $i\in [N]$.
 \For{$t = T_0+1, \ldots, T$,} 
 \State Observe user feature vector $q_t$
 \State Rotate the user feature $q_t' = \hat{V}^\top q_t$.
 \State Compute $x_{it,rtv}$, the truncated vectorization of $p_i'(q_t')^\top$ as in (\ref{eqn:rtv}).
 \State Compute $z_{it} = x_{it,rtv}^\top \widehat{\theta}_{t,rtv} + \beta_{it}$ as in (\ref{eqn:def_zit}).
 \State Select $S_t = \argmax_{S} \tilde{R}_t(S)$ via StaticMNL and observe choice $\mathbf{y}_t$.
 \State Update the MLE $\widehat{\theta}_{t,rtv}$ by solving (\ref{eqn:loss_rtv}).
 \EndFor
 \end{algorithmic}
\end{algorithm}

\begin{rmk}[Cold-start and Long-tail Issues] 
A common challenge in recommender systems is the \textit{cold-start problem} \citep{schein2002methods, lika2014facing}, where either new users or new items arrive with limited or no interaction history. Our framework addresses these scenarios through its contextual and low-rank structure. Each arriving user is represented by a feature vector, so the model does not depend on repeat arrivals. Similarly, when a new item is introduced, its feature vector $p_i$ can be directly incorporated into the bilinear utility model $p_i^\top \Theta^* q_t$, enabling preference estimation and uncertainty quantification without requiring past interaction data. This ensures robustness in environments with dynamically evolving user and item sets.  

For \textit{long-tail items} that may be chosen infrequently, the bilinear low-rank formulation helps mitigate data sparsity by sharing statistical strength across items and users. This allows utilities to be estimated reliably even when direct observations of a particular item are limited. Nonetheless, in cases where item or user features are partially unobserved, imputations based on historical data or clustering methods can be employed, though such procedures may introduce bias. Additional strategies, such as incorporating richer content-based features or adopting active exploration targeted at under-represented items, may further improve performance, and we view these extensions as promising directions for future research. 
\end{rmk}

\section{Theory} \label{sect:theory}

In this section, we first state and discuss the feasibility of the required assumptions (Section \ref{subsect:assumptions}). Next, we derive the estimation error bound on our rank-constrained maximum likelihood estimation for the LOAD problem (Proposition \ref{prop:lowrankest}), and use this result to establish a confidence bound on the parameter in the reduced space (Lemma \ref{lem:thetabound}). Finally we establish an upper bound on the regret of our ELSA-UCB policy (Theorem \ref{thm:regretbound}).

\subsection{Assumptions} \label{subsect:assumptions}

First, we impose conditions on the distribution of the feature vectors.

\begin{assm}[Distribution of user features] \label{assm:user_feature}
    Each user feature $q_t$ is i.i.d. from an unknown distribution, where $\|q_t\|_2 \le S_q$ for some constant $S_q$ and $\bbE [q_t q_t^\top] = \Sigma_q$ with positive definite $\Sigma_q$. Moreover, $\Sigma_q^{-1/2}q_t$ is $\sigma$-sub-Gaussian and its $\ell_2$ norm bounded by some constant $S_v$.
\end{assm}

The boundness of feature vectors is common in contextual bandit and assortment problems \citep{li2017provably, oh2021multinomial}, and the invertibility of the second moment ensures that the distribution of user features is sufficient to estimate the related parameters. In addition, note that we only require the i.i.d. condition on the user features, not on the joint user and item features as in the existing dynamic assortment literature \citep{oh2021multinomial, chen2020dynamic}.

\begin{assm}[Design of item features] \label{assm:item_feature}
    The item features satisfy $\|p_i\|_2 \le S_p$ for some constant $S_p$ and assume $\Sigma_p := \frac{1}{N} \sum_{i=1}^N p_i p_i^\top$ is positive definite.
\end{assm}

We assume that item features are also bounded with a positive definite second moment. Note that the positive definite condition is needed to guarantee the identifiability of the matrix parameter $\Phi^*$, and can be verified in practice since $p_i$'s are given in advance. Next, we impose a standard assumption in contextual MNL problems \citep{oh2021multinomial}, which is a slight modification of the assumption on the link function in generalized linear contextual bandits \citep{li2017provably}.

\begin{assm} \label{assm:kappa}
    There exists a constant $\kappa > 0$ such that
    $$
    \min_{\|\Phi - \Phi^*\|_F^2 \le 1} p_t(i|S, \Phi) p_t(0|S, \Phi) \ge \kappa,
    $$
    for every $S \in \calS$, $i \in S$ and $t\in [T]$, where $p_t(i|S,\Phi)$ is the probability of selecting item $i$ for user $t$ under the true parameter $\Phi$ and assortment $S$.
\end{assm}

This assumption ensures that the choice probabilities under the MNL model are “well-behaved” around the true parameter $\Phi^*$. 
Intuitively, this assumption prevents degeneracy in the choice probabilities. That is, it avoids situations where one item dominates the entire choice set (i.e., is almost always chosen), or where the no-purchase probability becomes negligible or overwhelming. Such extreme cases would make the model highly sensitive to noise and impede reliable estimation of user preferences. In MNL-bandit models, analogous assumptions have been used in prior work, including \citet{cheung2017thompson}, \citet{chen2020dynamic}, and \citet{oh2021multinomial}.

\subsection{Main Results}

First we establish an estimation error bound on our rank-constrained maximum likelihood estimation for the LOAD problem.

\begin{prop}[Accuracy of Low-Rank Estimation]\label{prop:lowrankest}
    Suppose Assumptions \ref{assm:user_feature}, \ref{assm:item_feature} and \ref{assm:kappa} hold. Further assume we have an initial estimate $\Phi_0\in \bbR^{d_2\times d_1}$ such that $\|\Phi_0 - \Phi^*\|_F \le c_2 \sqrt{\sigma_r}$ for some constant $c_2$, where $\sigma_r$ is $r$-th largest singular value of $\Phi^*$. Let $\hat{\Phi}$ be the rank-constrained estimator using Algorithm \ref{alg:fgd} with initial estimate $\Phi_0$. If the exploration length $T_0$ satisfies 
    $$
    T_0 \ge 4\sigma^2 (C_1\sqrt{d_1d_2} + C_2\sqrt{\log(1/\delta)})^2 + \frac{2B}{\lambda_{\min}(\Sigma)K} , 
    $$
    for some constants $C_1, C_2$ and any $\delta \in (0,1)$, $B>0$, where $\Sigma = \Sigma_q \otimes \Sigma_p$, then
    $$
    \|\hat{\Phi} - \Phi^*\|_F^2 \le \frac{288\sigma_1^2}{\sigma_r \kappa}\left( \frac{3}{4} + \frac{8T_0}{\kappa B} + \frac{2}{K\|\Sigma\|_2} \right) \frac{Kr\log(\frac{d_1+d_2}{\delta})}{B} S_2^4
    $$
    with probability at least $1-\delta$, where $S_2 = \max\{S_p, S_q\}$ is a constant.
\end{prop}

\begin{rmk}[Initial Estimate $\Phi_0$]
    Note that we require the initialization error $\|\Phi_0 - \Phi^*\|_F^2$ to be bounded. This assumption, often referred to as the `basin of attraction' is a standard practice in the non-convex low-rank matrix estimation literature \citep{jain2013low, chi2019nonconvex}. It ensures that the iterative algorithm converges to the desired target. Such assumptions are valid in diverse situations, such as when offline data are available, or when the length of the exploration stage is large enough to provide a warm start for the gradient descent algorithm.

    The initial estimator $\Phi_0$ can also be chosen as the unconstrained maximum likelihood estimator (MLE) based on data collected during the exploration phase. Under standard regularity conditions for the MNL model, the unconstrained MLE achieves the convergence rate: $\|\hat\Phi_{MLE} - \Phi^*\|_F = O(\sqrt{d_1d_2/T_0})$ \citep{van2000asymptotic}. Therefore, if the exploration length $T_0$ satisfies $T_0 \gtrsim d_1d_2/\sigma_r^{2}$, the unconstrained MLE will satisfy the initial condition. 
\end{rmk}

Note that the result holds for an arbitrary choice of $B>0$. Increasing $B$ and therefore increasing the length of exploration $T_0$ will naturally decrease the estimation error. In this case, accurate estimation of the subspace will lead to better exploitation. However, increasing $T_0$ will also increase the total cumulative regret. Hence the choice of $T_0$ shows the exploration-exploitation trade-off. In our method, we choose the optimal length of exploration $T_0 = O(\sqrt{T})$ to minimize the cumulative regret. Finally, note that the rate is proportional to the rank $r$, which would have been scaling with $\min\{d_1, d_2\}$ without the low-rank assumption. This illustrates how the low-rank assumption plays a crucial role in reducing the estimation error.

Next we first establish confidence bound on the parameter $\theta^*_{rtv}$.
\begin{lem}[Confidence bound on reduced space] \label{lem:thetabound}
    Suppose Assumptions \ref{assm:user_feature}, \ref{assm:item_feature} and \ref{assm:kappa} hold. Further assume that $\|\Theta^*_{r+1:d_2,r+1:d_1}\|_F^2 \le S_\perp$ for some $S_\perp$. Let $\widehat\theta_{n,rtv}$ be the maximum likelihood estimator of $\theta^*_{rtv}$ using observations from user $t=1, \ldots, n$ by minimizing (\ref{eqn:loss_rtv}). Then,
    \begin{align*}
        \|\widehat\theta_{n,rtv} - \theta^*_{rtv}\|_{W_n}
        &\le \frac{1}{\kappa nK} \left( \sigma \sqrt{2 df \log\left( 1 + \frac{n}{df} \right) + 4\log{n}} + S_2^2 \cdot S_\perp \cdot \sqrt{nK} \right) =: \alpha_n
    \end{align*}
    with probability at least $1-1/n^2$, where $W_n = \frac{1}{nK} \sum_{t=1}^n \sum_{i\in S_t} x_{it,rtv}x_{it,rtv}^\top$ and $df = (d_1+d_2)r - r^2$ is the dimension of the reduced space.
\end{lem}

With this definition of $\alpha_n$, we can calculate the optimistic utility defined in (\ref{eqn:def_zit}). Finally, we establish a non-asymptotic upper bound on the cumulative regret of our ELSA-UCB algorithm.

\begin{thm}[Regret Bound of ELSA-UCB]\label{thm:regretbound}
Suppose assumptions of Proposition \ref{prop:lowrankest} hold. Define
\begin{align*}
T_0 &:= 4\sigma^2 (C_1\sqrt{d_1d_2} + C_2\sqrt{\log{T}})^2 + \frac{2B}{\lambda_{\min}(\Sigma)K}, \\
S_{\perp} &:= \frac{288\sigma_1^3}{\sigma_r^3 \cdot \kappa} \cdot \left( \frac{3}{4} + \frac{8T_0}{\kappa B} + \frac{2}{K\|\Sigma\|_2} \right) \frac{Kr\log(\frac{d_1+d_2}{\delta})}{B} S_2^4, \\
\alpha_n &:= \frac{1}{\kappa} \left( \sigma \sqrt{2 df \log\left( 1 + \frac{n}{df} \right) + 4\log{n}} + S_2^2 \cdot S_\perp \cdot \sqrt{nK} \right),
\end{align*}
for any $B>0$, where $df = (d_1+d_2)r - r^2$ is the degree of freedom on the reduced space. Then applying Algorithm \ref{alg:ELSA-UCB} with $T_0$ and $\alpha_n$, the cumulative regret of our ELSA-UCB policy satisfies
$$
{\calR}_T 
\le R_{\max} \left( T_0 + 3 + 2 \alpha_T \sqrt{T \cdot 2\cdot df \log(T/df)} + 2S_2^2 S_\perp(T-T_0) \right),
$$
where $R_{\max} = \max_{i\in [N]} r_i$.
\end{thm}

Theorem \ref{thm:regretbound} provides a detailed regret bound of our policy. When applying our policy to e-commerce applications with a large number of users (large $T$), we can further simply this rate. With the choice $B = b\sqrt{T}$ for some constant $b$, we have $S_\perp = O(r/\sqrt{T})$ by definition and therefore the sub-linear regret $\calR_T = \tilde{O}(r(d_1+d_2)\sqrt{T})$, where $\tilde{O}(\cdot)$ ignores some logarithm term, which implies that the mean regret converges to $0$ as $T\to \infty$. For $d$-dimensional contextual assortment problem, \citet{chen2020dynamic} established an $\Omega(d\sqrt{T}/K)$ lower bound on the cumulative regret. As the capacity $K$ is the maximum number of items that are displayed each time, we can interpret $K$ as a constant, and consequently, the lower bound becomes $\Omega(d\sqrt{T})$. When these results are applied by interpreting vectorization of the outer product of dual features as a single feature, the lower bound becomes $\Omega(d_1d_2\sqrt{T})$, without the low-rank assumption. Our algorithm significantly reduces the rate of regret to $\tilde{O}(df \cdot \sqrt{T})$ in the low-rank regime.

Establishing this regret bound possesses two main technical challenges. The first challenge involves incorporating the low-rank structure. To address this, we utilize results from the low-rank matrix estimation literature and adapt them to the LOAD problem. This enables us to establish an estimation bound on the estimated subspace. The second challenge arises from the inherent combinatorial structure of the assortment problem. To overcome this, we construct bound on the difference between the rewards in the original problem and the transformed problem in the reduced space, which enables us to establish regret bounds for our algorithm. We provide a proof sketch in Section S.7.1 of the Appendix.

\section{Numerical Studies}

In this section, we evaluate the cumulative regret  of our ELSA-GIC policy across varying scenarios, including different numbers of users, feature dimensions, ranks of parameter matrix, numbers of items and capacity of the assortment. We begin by  examining it using synthetic data in Section \ref{subsect:simulation} and subsequently validate its practicability in a real-world dataset in Section \ref{subsect:realdata}.  

\subsection{Simulation}\label{subsect:simulation}

To generate the parameter matrix $\Phi^*$, we first construct a diagonal matrix $D^*$ with $r$ non-zero components, each set to be 10. Next we multiply random orthonormal matrices $U^*\in \bbR^{d_2\times r}$ and $V^*\in \bbR^{d_1\times r}$ to generate a random rank-$r$ matrix $\Phi^* = U^*D^*(V^*)^\top \in \bbR^{d_2\times d_1}$. 
For the user features, we generate $q_t \in \bbR^{d_1}$ for every $t \in [T]$ from the $d_1$-dimensional standard Gaussian distribution. Similarly, for item features, we generate $p_i \in \bbR^{d_2}$ from the $d_2$-dimensional standard Gaussian distribution for every $i \in [N]$. Note that the $N$ item features are fixed over the time horizon while user features are generated at each time step. The revenue is set as $r_i = 1$ for every $i \in [N]$, following the setup in \citet{oh2021multinomial}. This accounts for scenarios where the goal is to maximize the sum of binary outcomes, aiming to maximize the total number of conversions, such as clicks or purchases in online retail.

Utilizing the generated parameters and features, we simulate user choices and the corresponding assortment selections following the policy. We compare the performance of our proposed ELSA-UCB with the estimated rank from GIC (denoted ELSA-GIC) against two versions of the UCB-MNL \citep{agrawal2019mnl, oh2021multinomial}. As UCB-MNL considers features for each user-item pair as $x_{it}$, we consider two formulations of the joint features; ``stacked'' and ``vectorized''.

\begin{itemize}
    \item \textbf{Stacked UCB-MNL}: As shown in Figure \ref{fig:StackUCBMNL_v}, we stack the user features and the item features together including intercept as $x_{it} = (1, p_i, q_t) = (1, p_{i1}, p_{i2}, \ldots, p_{id_2}, q_{t1}, \ldots q_{td_1}) \in {\bbR}^{d_1+d_2+1}$,
    which allows the linear formulation of the utilities to take into account the main effects of both features with parameter $\theta = (\alpha, \beta_1, \ldots, \beta_{d_2}, \gamma_1, \ldots, \gamma_{d_1}) \in \bbR^{d_1 + d_2 + 1}$ :
    $$
    v_{it} = \theta^\top x_{it} = \alpha + \sum_{j=1}^{d_2} p_{ij} \beta_j + \sum_{k=1}^{d_1}q_{tk} \gamma_k. 
    $$

    \item \textbf{Vectorized UCB-MNL}: As shown in Figure \ref{fig:VectUCBMNL_v}, we consider the vectorization of the outer product of the dual features with intercept, $p_i' = (1, p_{i1}, \ldots p_{id_2})$ and $q_t' = (1, q_{t1}, \ldots, q_{td_1})$: 
    \begin{align*}
    X_{it} &= p_i' \otimes q_t' = \begin{pmatrix}
        1 & q_{t1} & \ldots & q_{td_1} \\
        p_{i1} & p_{i1} q_{t1} & \ldots & p_{i1} q_{td_1} \\
        \vdots & \vdots & \ddots & \vdots \\
        p_{id_2} & p_{id_2} q_{t1} & \ldots & p_{id_2} q_{td_1}
    \end{pmatrix} \in {\bbR}^{(d_2+1)\times (d_1+1)},\\
    x_{it} &= \vect(X_{it}) = (1, p_{i1}, \ldots, p_{id_2} q_{td_1}) \in {\bbR}^{(d_2+1) (d_1+1)},
    \end{align*}
    as the joint feature for the user-item pair. This takes into account for the main effects and the interaction terms of the dual features with parameter $\theta = (\alpha, \ldots, \delta_{11}, \ldots, \delta_{d_2d_1}) \in \bbR^{(d_2+1)(d_1+1)}$, which leads to the following formulation of utilities:
    $$
    v_{it} = \theta^\top x_{it} = \alpha + \sum_{j=1}^{d_2} p_{ij} \beta_j + \sum_{k=1}^{d_1}q_{tk} \gamma_k + \sum_{j=1}^{d_2}\sum_{k=1}^{d_1} \delta_{jk} p_{ij} q_{tk}. 
    $$
\end{itemize}
Note that the stacked formulation can be viewed as a special case of the vectorized formulation, when there is no interaction. In this case, the parameter matrix contains non-zero elements only on the first row and the first column, which restricts the rank of the matrix to be at most $2$.

We apply our proposed ELSA-GIC policy and the two variations of UCB-MNL as baselines for comparison. We repeat the simulation 50 times and plot the mean and 95\% confidence interval of the cumulative regret for each $T\in \{1000, 2000, 3000, 5000, 7500, 10000\}$ under different environments.

First we demonstrate the consistency of rank-estimation through GIC (\ref{eqn:GIC}). We fix $d_1 = 30, d_2 = 20, N = 20, K = 3, r = 3$ and change the number of users over $T = 100, 200, 500, 1000, 2000$ to compare the estimated rank when $T$ increases. As shown in Figure \ref{fig:rankest}, the probability of accurate rank estimation nearly converges to $1$ as the sample size increases.
\begin{figure}[h]
    \centering
    \includegraphics[width = 0.5 \textwidth]{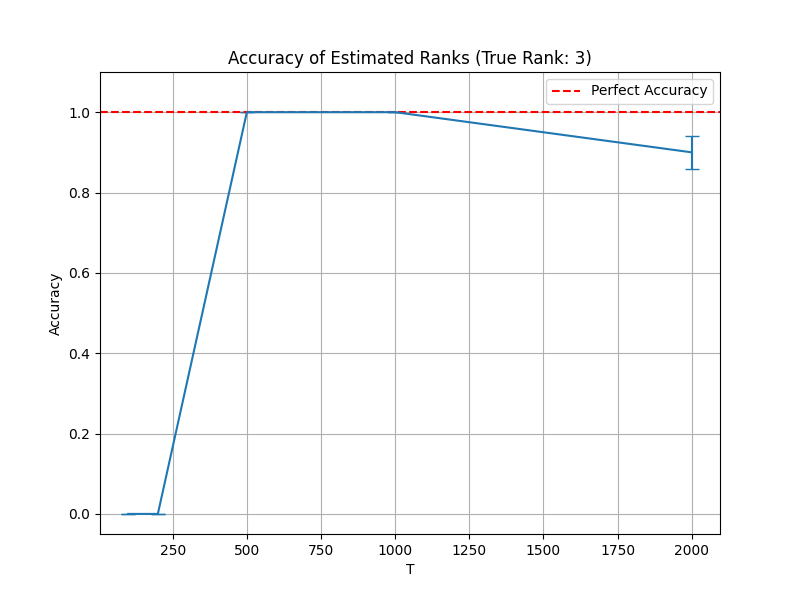}
    \caption{Accuracy of estimated rank using GIC with increasing $T$.}
    \label{fig:rankest}
\end{figure}

\begin{figure}[h]
    \centering
    \includegraphics[width = 1.0 \textwidth]{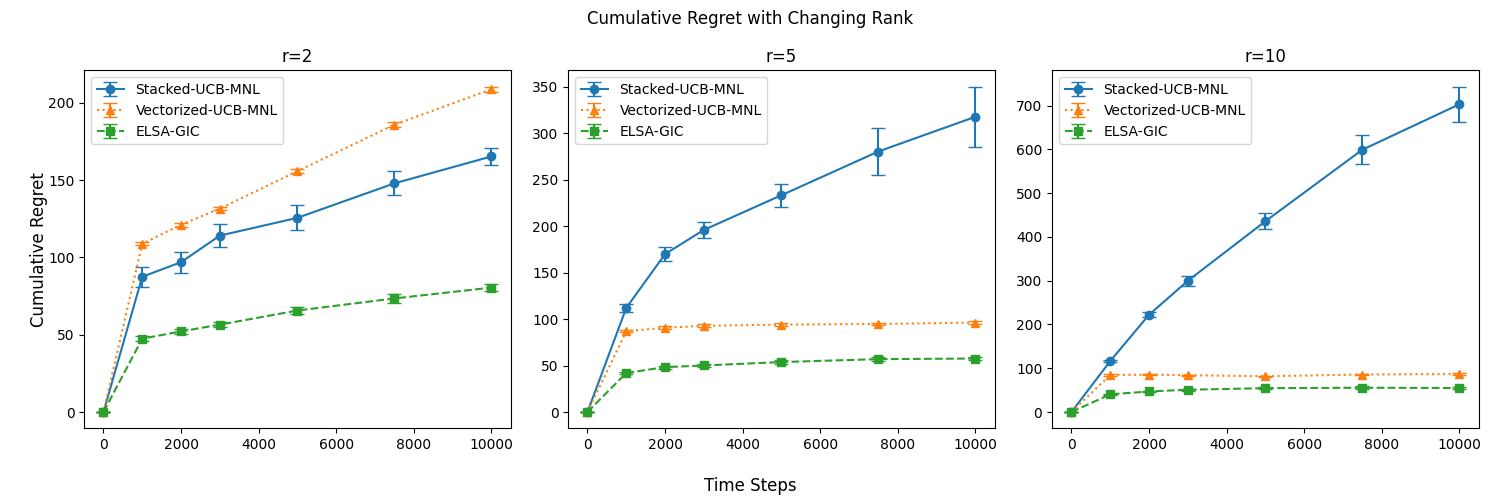}
    \caption{Cumulative regrets of our proposed ELSA-GIC algorithm compared with Stacked UCB-MNL and Vectorized UCB-MNL for different ranks $r\in \{2,5,10\}$.}
    \label{fig:change_r}
\end{figure}

For the first line of experiments, to assess the performance in terms of cumulative regret, we fix $d_1 = 50, d_2 = 20, N = 20, K = 3$ and change the rank $r \in \{2,5,10\}$ and number of users over $T =1000, 2000, 3000, 5000, 7500, 10000$ and plot the corresponding cumulative regret in Figure \ref{fig:change_r}. In higher ranks, the Stacked UCB-MNL produces larger cumulative regret, as it does not take consideration for the diverse interactions between the dual features. On the other hand, the Vectorized UCB-MNL has comparable performance when the rank is high. However, when the rank is low, the performance of Vectorized UCB-MNL is not satisfactory. Our ELSA-GIC policy universally outperforms both methods as it utilizes the low-rank structure and therefore efficiently estimate how the interactions between the dual features affect user preference of items. We also note that the Stacked method corresponds to a special case with $r=2$, which is only correctly specified when the true parameter matrix has nonzero entries solely in the main-effect terms, the first row and the first column. In our simulations, however, the parameter matrix includes nonzero interaction terms, so the Stacked method is mis-specified even when the true rank equals $2$. Our method, in contrast, remains correctly specified in this case and achieves superior performance.

Note that the cumulative regrets across different settings are not directly comparable. For example, in Figure \ref{fig:change_r}, we observe a decrease in cumulative regrets for our method as the rank increases. This phenomenon occurs as we fix the scale for the singular values of $\Phi^*$, the intrinsic value $v_{it}$ increases, making it be easier to distinguish the optimal assortments. However, if we decrease the scale of the singular values as rank increases, the scale of the intrinsic value might remain the same, but the parameter estimation becomes challenging, potentially leading to performance degradation. Therefore, it would be fair to compare cumulative regret only within the same settings and compare the relative gap between different settings.

\begin{figure}[h]
    \centering
    \includegraphics[width = 1.0 \textwidth]{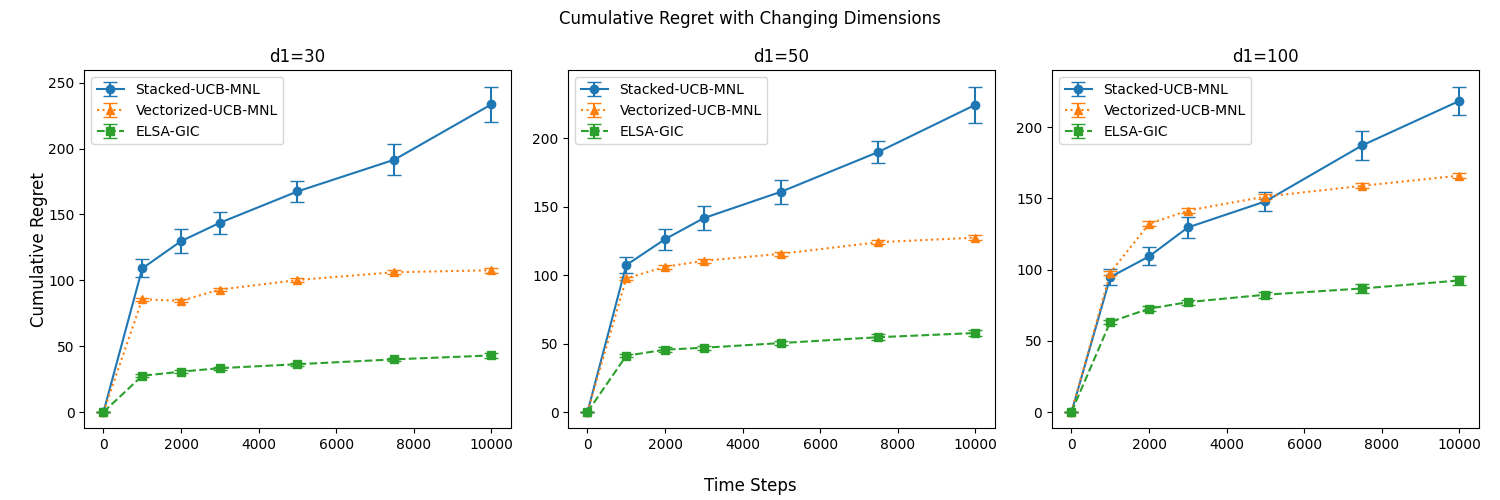}
    \caption{Plots of cumulative regret by $T$ for different dimensions.}
    \label{fig:change_d1}
\end{figure}

In the next set of experiments, we fix $ d_2 = 20, r = 3, N = 20, K = 3$ while changing $d_1 \in \{30, 50, 100\}$. The results are shown in Figure \ref{fig:change_d1}. The Stacked UCB-MNL shows poor performance across all dimensions. Moreover, we can observe that the gap between the cumulative regret of ELSA-GIC and Vectorized UCB-MNL widens as $d_1$ increases. This indicates that our algorithm is even preferable in environments with higher dimensional features. Furthermore, note that the scale of the regret increases as $d_1$ increases which aligns with the results of Theorem \ref{thm:regretbound}.

\begin{figure}[h]
    \centering
    \includegraphics[width = 1.0 \textwidth]{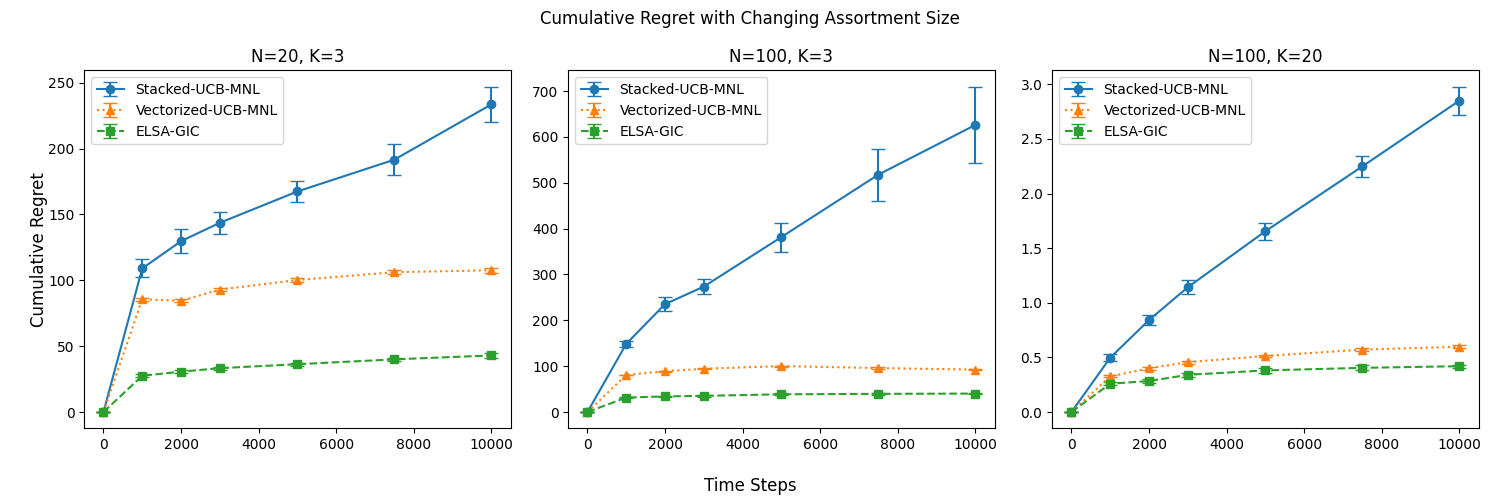}
    \caption{Plots of cumulative regret by $T$ for different numbers of items and capacity.}
    \label{fig:change_N}
\end{figure}

In the third set of experiments, we fix $d_1 = 30, d_2 = 20, r = 3$ and change $(N,K)$ to be $\{(20,3), (100,3), (100,20)\}$. The results are illustrated in Figure \ref{fig:change_N}. Note that as the assortment capacity $K$ increases, the gap between the optimal assortment and the near-optimal assortments decreases, resulting in smaller regrets. Overall, our proposed method outperforms both Stacked UCB-MNL and Vectorized UCB-MNL across all setups.

In the final set of experiments, we evaluate the robustness of ELSA-GIC under model misspecification. We consider three scenarios:
\begin{enumerate}
    \item \textbf{Approximately low-rank $\Phi^*$}: The first $r$ singular values of $\Phi^*$ are fixed at a constant level, while the remaining singular values are set to be small  - specifically, 1\% of the leading singular values - capturing approximate low-rank structure with residual noise.
    \item \textbf{Full-rank $\Phi^*$}: The parameter matrix is drawn as a full-rank matrix without dominant singular directions, favoring the Vectorized UCB-MNL model, which does not assume any low-rank structure.
    \item \textbf{Main-effect-only model}: The utility is generated without any interaction between user and item features and is influenced solely by the main effects. In this setting, the effective rank is 2, and the structure favors the Stacked UCB-MNL model.
\end{enumerate}

We fix $(d_1, d_2, N, K) = (30, 20, 20, 3)$, and use $r=3$ for the approximately low-rank matrix. As shown in Figure \ref{fig:misspecified}, ELSA-GIC achieves the lowest regret in the approximately low-rank case, demonstrating its strength when the low-rank assumption only holds approximately. Moreover, it remains competitive even when the data-generating process favors alternative models, achieving comparable regret to the specialized methods under their respective advantageous regimes.

\begin{figure}[h]
    \centering
    \includegraphics[width = 1.0 \textwidth]{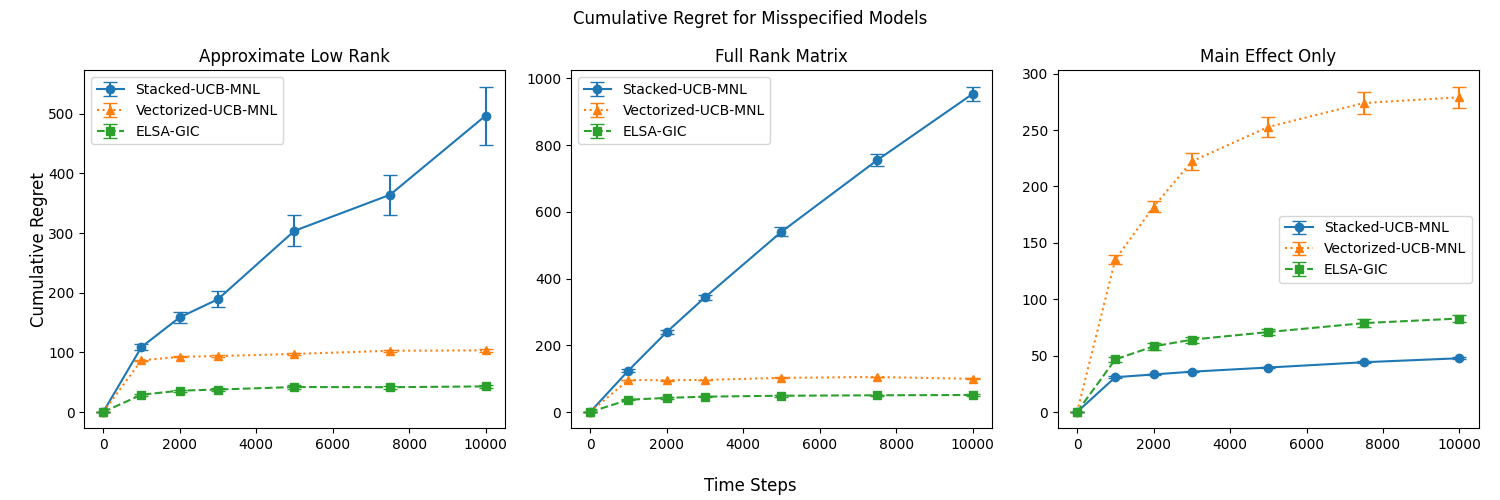}
    \caption{Cumulative regret over time $T$ under various model misspecification settings.}
    \label{fig:misspecified}
\end{figure}

\subsection{Real Data Application} \label{subsect:realdata}

In this section, we apply our algorithm to the ``Expedia Hotel" dataset \citep{expedia-personalized-sort} to evaluate its performance in real data.

\subsubsection{Data Pre-processing}

The dataset consists of 399,344 unique searches on 23,715 search destinations each accompanied by a recommendation of maximum 38 hotels from a pool of 136,886 unique properties. User responses are indicated by clicks or hotel room purchases. The dataset also includes features for each property-user pair, including hotel characteristics such as star ratings and location attractiveness, as well as user attributes such as average hotel star rating and prices from past booking history.

It is important to note that search queries constrain assortments to those within the search destination. We focus on data from the top destination with the highest search volume. Additionally, we filter out columns with missing over $90\%$ of searches. To address missing values in the remaining data for average star ratings and prices of each user from historical data, we employ a simple regression tree imputation. For hotel features, we compute the mean of feature values across searches. We discretize numerical features such as star ratings or review scores into categories to improve model fitting. After pre-processing, we have $T = 4465$ unique searches encompassing $N = 124$ different hotels, with $d_2=10$ hotel features and $d_1=18$ user features. The description of features is given in Table \ref{table:features}. We normalize each feature to have mean $0$ and variance $1$. As for the item features, we truncate the feature values to lie within $[-3,3]$ to avoid the effect of outliers. 

\begin{center}
\begin{table}[]
\begin{tabular}{l|l}
\hline
Item (Hotel) Features  & \begin{tabular}[c]{@{}l@{}}Star rating (4 levels), Average review score (3 levels), Brand \\of property,  Location score, Log historical price, Current price,\\ Promotion\end{tabular}                                                                 \\ \hline
User (Search) Features & \begin{tabular}[c]{@{}l@{}}Length of stay (4 levels), Time until actual stay (4 levels), \\ \# of adults (3 levels), \# of children (3 levels), \# of rooms\\ (2 levels), Inclusion of Saturday night, Website domain, User\\ historical star rating (3 levels), User historical price (4 levels)
\end{tabular} \\ \hline
\end{tabular}
\caption{Table of item and user features.}
\label{table:features}
\end{table}
\end{center}
\vspace{-2em}

\subsubsection{Analysis of the Expedia Dataset}

To evaluate different dynamic assortment selection policies, we first estimate the ground truth parameter $\Phi^*$ using rank-constrained maximum likelihood estimation (\ref{eqn:rank_constrained_MLE}) on the full dataset. 
Consequently, the Generalized Information Criterion (GIC) selects a rank of $r=5$ across a range of candidate ranks. Following the procedure used in our simulation studies, we compare the cumulative regrets of our proposed ELSA-GIC policy with rank estimation against Stacked UCB-MNL and Vectorized UCB-MNL. We evaluate the algorithms at different time horizons $T\in \{2000, 5000, 10000\}$. To assess the performance of the algorithm over an extended time horizon, we resample the users based on the original dataset.

\begin{center}
\begin{table}[h!]
\begin{center}
\begin{tabular}{c|c|c}
\hline
Number of Samples & vs. Stacked-UCB-MNL & vs.  Vect-UCB-MNL  \\ \hline
2000  &  26.3 \% ($\pm$ 1.0 \%) & 2.5 \% ($\pm$ 0.7 \%)   \\ \hline
5000  & 45.8 \% ($\pm$ 1.0 \%) &  3.1 \% ($\pm$ 1.0 \%)  \\ \hline
10000  & 61.8 \% ($\pm$ 0.8 \%)  &  3.5 \% ($\pm$ 1.4 \%) \\ \hline
\end{tabular}
\caption{Performance improvement of our algorithm over baseline methods.}
\label{table:expedia_results}
\end{center}
\end{table}
\end{center}
\vspace{-2em}

As shown in Table \ref{table:expedia_results}, our proposed ELSA-GIC algorithm outperforms both Stacked UCB-MNL and Vectorized UCB-MNL in terms of cumulative regret. Notably, the performance gap between ELSA-GIC and the baseline methods widen as the time horizon increases. Even though the real dataset does not exhibit a strictly low-rank structure relative to its dimensions ($r=5$ vs. $(d_1, d_2) = (18,10)$), our method still achieves superior performance. By avoiding unnecessary full-dimensional estimation and instead concentrating learning within a reduced subspace, ELSA-GIC effectively improves the cumulative reward compared to the baseline methods.

\section*{Acknowledgments}
The authors would like to thank the editor, associate editor, and reviewers for their insightful comments, which have greatly improved the presentation. The research was partially supported by NSF grant SES 2217440.

\newpage

\appendix 

\baselineskip=24pt
\setcounter{page}{1}
\setcounter{equation}{0}
\setcounter{section}{0}
\renewcommand{\thesection}{S.\arabic{section}}
\renewcommand{\theequation}{S\arabic{equation}}

\begin{center}
{\Large\bf Supplementary Materials} \\
\medskip
{\Large\bf ``Low-Rank Online Dynamic Assortment with Dual Contextual Information"}  \\
\bigskip
\end{center}
\bigskip

\noindent
In this supplement, we first discuss and compare additional related works in Section \ref{sect:Appendix_related_Work}. We provide high probability regret bound in \ref{sect:Appendix_hpbound}. Next, we provide sensitivity analysis regarding the hyperparameters of UCB in Section \ref{sect:Appendix_sensitivity}. We also discuss the implementation of UCB-MNL, the baseline algorithm for comparison in \ref{sect:Appendix_UCBMNL}, and provide the detailed rank tuning implementation and the batch extension of our algorithm in Section \ref{sect:AppendixAlgorithm}. A few interesting future directions are provided in Section \ref{sect:Appendix_future_work}. Finally, we include detailed proofs of the theorems and lemmas in Section \ref{sect:AppendixProof}. 

\section{Additional Related Work} \label{sect:Appendix_related_Work}

In this section, we introduce and discuss the distinction with related works in the literature on low-rank bandits, bandits with sparsity assumption and recommendation systems.

\subsection{Comparison with Low-rank Bandits}

Here, we highlight the distinction between our work and existing literature in bandit problems.

\begin{enumerate}
    \item \textbf{Theoretical distinction from bilinear bandits}:
    We now clearly articulate the key differences between our setting and the bilinear bandit literature (e.g., \citet{jun2019bilinear}) in both theory and methodology. In particular, we highlight how the MNL feedback model, the combinatorial action space of assortments, and the non-linear reward structure lead to substantial differences in problem formulation, analysis, and algorithm design. These distinctions are now discussed in detail in our response to your next comment.
    \item \textbf{Algorithmic structure and adaptation}:
    While the factored gradient descent (FGD) algorithm we adopt for subspace estimation is indeed rooted in the low-rank matrix literature, our analysis tailors it to the categorical MNL feedback setting and proves convergence properties specific to this structure. Furthermore, our ELSA-UCB framework departs from existing two-stage bandits by incorporating a confidence-adjusted refinement over a learned low-rank subspace with categorical choice-based feedback. This leads to a new regret decomposition and necessitates novel bounding techniques.
\end{enumerate}

\subsection{Comparison with Sparsity Assumption}
Both low-rank and sparse structures are widely used to address high dimensionality in statistical learning. Prior works such as \citet{udell2019big}, \citet{koren2009matrix}, and \citet{bell2007lessons} support the presence of low-rank structures in high-dimensional settings, particularly in recommendation systems. On the other hand, sparsity assumptions have also been effectively employed in contextual dynamic assortment optimization, as seen in \citet{wang2019online} and \citet{shao2022sparse}, where LASSO-based methods are used to reduce sample complexity. In this work, we focus on the low-rank structure for two main reasons:

\begin{itemize}
    \item \textbf{Matrix structure and bilinear modeling}: 
    A sparsity assumption on the parameter matrix $\Phi^*$ is equivalent to a sparsity assumption on the vectorized parameter matrix $\vect(\Phi^*)$. This facilitates the use of standard sparse bandit techniques. In contrast, the low-rank assumption explicitly leverages the matrix structure of $\Phi^*$. This structural difference motivates the development of new theory and algorithms tailored to low-rank structure, as undertaken in our work.
    \item \textbf{Interpretability through latent subspaces}:
    The low-rank assumption enables learning of intrinsic latent representations of both user and item features. In particular, it reveals low-dimensional subspaces that drive user–item interactions, offering a more interpretable and compact model of preference formation. In contrast, sparsity implies that only a small number of feature pairs affect preference, without identifying coherent subspace structure. This difference is particularly important in dual-context environments, where subspace learning can enhance both prediction and generalization.
\end{itemize}

\subsection{Recommendation Systems}
Recommendation system is a field with extensive research closely related to the dynamic assortment selection problem. These systems collect user preferences across a range of items and forecast these preferences. Multiple methodologies were developed for recommendation systems, including content-based filtering, demographic filtering, and collaborative filtering methods \citep{bobadilla2013recommender, dai2019smooth, dai2021scalable}. Content-based filtering tailors personalized recommendations by drawing insight from user's history, suggesting similar items based on their previous choices \citep{van2000using}. In contrast, collaborative filtering observes user ratings on items and then offers recommendations for each user by drawing data from similar users. To mitigate issues arising from the sparsity of observations, dimension reduction techniques including low-rank matrix factorization \citep{sarwar2000application} and covariate-assisted matrix completion \citep{robin2018low, mao2019matrix, ibriga2023covariate} have been applied to learn user preferences. However, in the dynamic assortment problem, the decision maker can only observe binary outcomes of user choices, which are limited to at most one selection per user, without any quantitative observations such as user ratings. This line of literature motivates us to investigate the dynamic assortment problem, involving user-item interactions to model user preferences toward items.

\section{High Probability Bound of Regret}\label{sect:Appendix_hpbound}
Here we provide results involving high probability bounds rather than bounds on the expectation of the regret in Corollary \ref{cor:regretbound_hp}.
\begin{cor}[High Probability Bound of Regret]\label{cor:regretbound_hp}
Suppose assumptions of Proposition 1 hold. Define
\begin{align*}
T_0 &:= 4\sigma^2 (C_1\sqrt{d_1d_2} + C_2\sqrt{\log{T}})^2 + \frac{2B}{\lambda_{\min}(\Sigma)K}, \\
S_{\perp} &:= \frac{288\sigma_1^3}{\sigma_r^3 \cdot \kappa} \cdot \left( \frac{3}{4} + \frac{8T_0}{\kappa B} + \frac{2}{K\|\Sigma\|_2} \right) \frac{Kr\log(\frac{d_1+d_2}{\delta})}{B} S_2^4, \\
\alpha_n &:= \frac{1}{\kappa} \left( \sigma \sqrt{2 df \log\left( 1 + \frac{n}{df} \right) + 4\log{n}} + S_2^2 \cdot S_\perp \cdot \sqrt{nK} \right),
\end{align*}
for any $B>0$, where $df = (d_1+d_2)r - r^2$ is the degree of freedom on the reduced space. Then by applying Algorithm 2 with $T_0 > 1/\delta$ and $\alpha_n$, the cumulative regret of our ELSA-UCB policy satisfies
$$
\sum_{t=1}^{T} r_t
\le R_{\max} \left( T_0 + 2 \alpha_T \sqrt{T \cdot 2\cdot df \log(T/df)} + 2S_2^2 S_\perp(T-T_0) \right),
$$
with probability at least $1-2\delta$, where $R_{\max} = \max_{i\in [N]} r_i$.
\end{cor}
The proof is provided in Appendix \ref{sect:AppendixProof}.

\section{Sensitivity Analysis of Hyperparameters}
\label{sect:Appendix_sensitivity}

\subsection{Confidence Bound Related Hyperparameters}
In this section, we provide a detailed description of the tuning parameter choices for the ELSA-UCB algorithm as well as implementation details for the Stacked UCB-MNL and Vectorized UCB-MNL baselines. The confidence radius in ELSA-UCB is based on Lemma 1 and is defined as follows:
$$
\|\widehat\theta_{n,rtv} - \theta^*_{rtv}\|_{W_n}
\le \frac{1}{\kappa nK} \left( \sigma \sqrt{2 df \log\left( 1 + \frac{n}{df} \right) + 4\log{n}} + S_2^2 \cdot S_\perp \cdot \sqrt{nK} \right) =: \alpha_n.
$$

To simplify tuning, we reparametrize $\alpha_n$ as:
$$
\alpha_n = \alpha\sqrt{2df \log(1+n/df) + 4\log{n}} + \beta\sqrt{nK},
$$
where $\alpha, \beta$ are tuning parameters. The first term mirrors the confidence radius used in UCB-MNL from \citet{oh2021multinomial}, and we use the same value across different methods. The second term accounts for the approximation error due to truncation. We perform a grid search over different values of $(\alpha, \beta)$ with $d_1 = 50, d_2 = 20, r=3$ to assess sensitivity and find the algorithm to be robust across a range of parameters. First we fix $\beta = 1$ and compare the regrets over $\alpha \in \{2, 5, 10\}$, as shown in the left plot of Figure \ref{fig:tuning}. Next, we compare the regrets over $\beta \in \{0.5, 1.0, 2.0\}$, as shown in the right plot of Figure \ref{fig:tuning}. We can verify that the performance of our algorithm is robust over the selection of hyperparameters. The final values $(\alpha, \beta) = (5,1)$ are used in all reported simulations. 

\begin{figure}[h!]
    \centering
    \includegraphics[width = 0.7\textwidth]{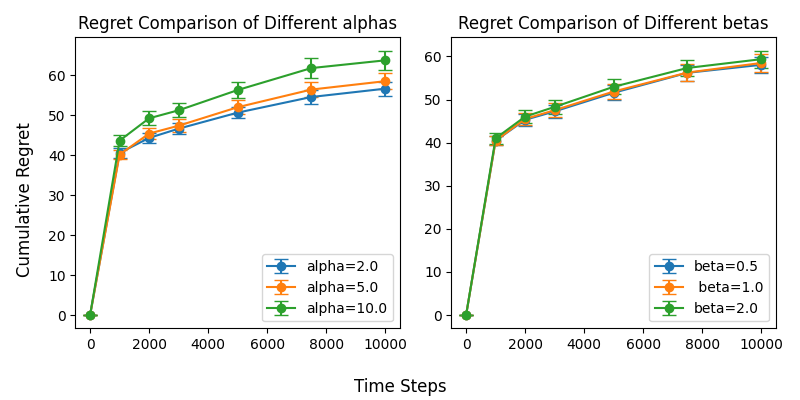}
    \caption{Cumulative regrets for different $\alpha,\beta$.}
    \vspace{-1em}
    \label{fig:tuning}
\end{figure}

\subsection{Exploration Length Related Hyperparameters}

As discussed in Theorem 1, the choice of the initial exploration length $T_0 = c d_1 d_2 + \sqrt{T}$ balances the trade-off between obtaining a reliable estimate of the low-rank structure and minimizing regret accumulated during random exploration. In our simulations, we fix $c = 0.2$. To examine sensitivity to this parameter, we conducted simulations varying the constant $c \in \{0.1, 0.2, 0.5\}$.

The results, shown in Figure~\ref{fig:change_c1}, demonstrate that our method achieves consistently low regret across a range of values of $c$. This confirms that the theoretical guideline provides a robust and practical choice for $T_0$.  

\begin{figure}[H]
    \centering
    \includegraphics[width=0.9\textwidth]{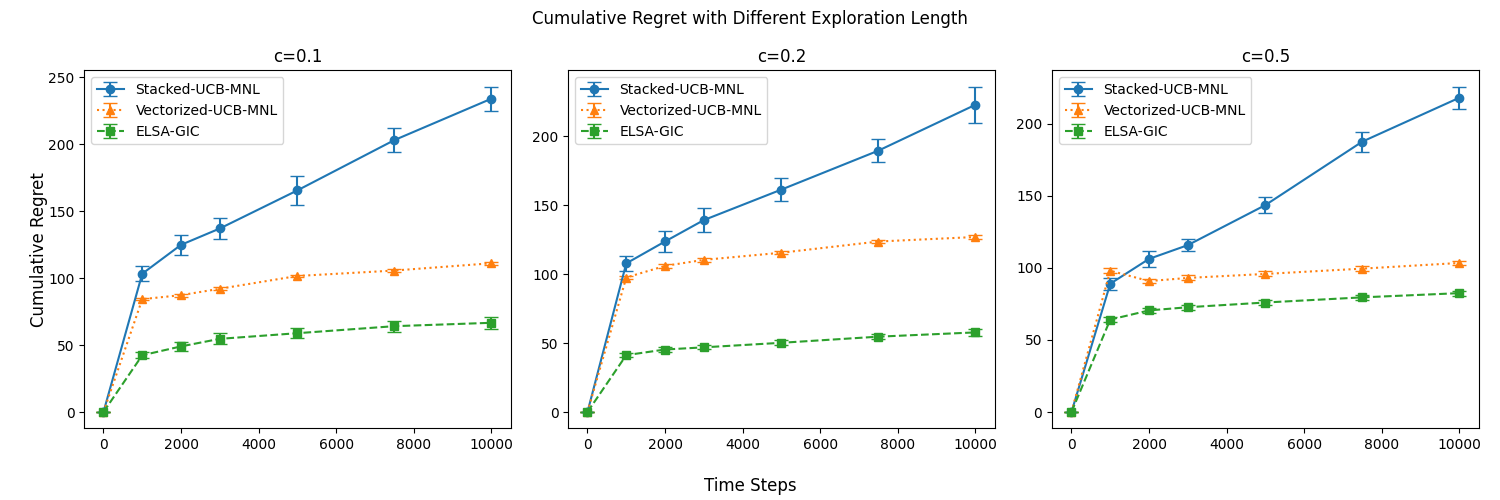}
    \caption{Cumulative regret by $T$ for different choices of the exploration parameter $c$.}
    \label{fig:change_c1}
\end{figure}

\section{Implementation of UCB-MNL} \label{sect:Appendix_UCBMNL}
In this section, we introduce UCB-MNL from \citet{oh2021multinomial} and how it is implemented for Stacked UCB-MNL and Vectorized UCB-MNL, as shown in Algorithm \ref{alg:UCB-MNL}.

\begin{algorithm}
 \caption{UCB-MNL \citep{oh2021multinomial}}
 \label{alg:UCB-MNL}
 \begin{algorithmic}[1]
 \Require Initialization $T_0$, confidence radius $\alpha_t$.
 \State For $t\in [T_0]$
 \State Randomly choose $S_t$ with $|S_t| = K$
 \State $V_t \leftarrow V_{t-1} + \sum_{i\in S_t} x_{it} x_{it}^\top$
 \For $t=T_0 + 1, \ldots, T$, 
 \State Compute $z_{it} = x_{it}^\top \hat\theta_{t-1} + \alpha_t \|x_{it}\|_{V_{t-1}^{-1}}$ for all $i$
 \State Offer $S_t = \argmax_{S} \tilde{R}_t(S)$ and observe $y_t$
 \State Update $V_t \leftarrow V_{t-1} + \sum_{i\in S_t} x_{it} x_{it}^\top$
 \State Compute MLE $\hat\theta_t$
 \EndFor
 \end{algorithmic}
\end{algorithm}

The notations $\tilde{R}_t(S)$ and $\alpha_t$ in Algorithm \ref{alg:UCB-MNL} are define as:
$$
\tilde{R}_t(S) := \frac{\sum_{i\in S} r_{it} \exp(z_{it})}{1 + \sum_{j\in S} \exp(z_{jt})},
$$
and $\alpha_t = \frac{1}{2\kappa} \sqrt{2d\log(1+\frac{t}{d}) + 2\log{t}}$.

The only difference between the Stacked and Vectorized versions lies in the construction of the context vector $x_{it}$: 
\begin{itemize}
    \item For Vectorized UCB-MNL, we use $x_{it} = \vect(p_iq_t^\top)$.
    \item For Stacked UCB-MNL, we use $x_{it} = (p_i^\top, q_t^\top)^\top$, i.e., the concatenation of the two context vectors.
\end{itemize}

\section{Additional Algorithms} \label{sect:AppendixAlgorithm}

In this section, we provide details of the sub-algorithm STATICMNL from \citet{rusmevichientong2010dynamic}, our ELSA-UCB algorithm with rank estimation, and the batch version of ELSA-UCB.

\subsection{STATICMNL}

Here we discuss and provide details of STATICMNL in Algorithm \ref{alg:STATICMNL}. The STATICMNL algorithm is used to compute the optimal assortment under the Multinomial Logit (MNL) model, given a fixed utility vector $\mathbf{v} = (v_1, \ldots, v_N)$ and reward vector $\br = (r_1, \ldots, r_N)$, subject to a capacity constraint $K$. Instead of exhaustively enumerating all $\sum_{k=1}^K {N\choose k}$ possible subsets (which is computationally heavy), the algorithm identifies a much smaller set of candidate assortments—specifically, $O(N^2)$ candidates, by iteratively examining intersection points that determine changes in the ranking of marginal value contributions.

The algorithm involves sorting utilities and intersection points, and updating candidate assortments accordingly. The resulting computational complexity is $O(N^2)$, ignoring logarithmic factors, which makes STATICMNL a computationally efficient subroutine for solving the inner maximization in our UCB-based stage. It is also worth noting that the algorithm does not have a dependency on $K$.

\begin{algorithm}
 \caption{STATICMNL}
 \label{alg:STATICMNL}
 \begin{algorithmic}[1]
 \Require Assortment capacity $K$, utility vector $v = (v_1, \ldots, v_N)$, reward vector $r = (r_1, \ldots, r_N)$.
 \State For each $j\in [N]$, define $I(0,j) = r_j$ and for $i,j \in [N]$ with $i\ne j$, define $I(i,j) = (v_iw_i - v_jw_j)/(v_i-v_j)$.
 \State Denote $L = N(N+1)/2$. Sort the pairs to $(i_1, j_1), \ldots(i_L, j_L)$ so that $0\le i_l < j_l \le N$ for every $l\in [L]$ and
 $$
 -\infty \equiv I(i_0, j_0) < I(i_1, j_1) \le \cdots \le I(i_L, j_L) < I(i_{L+1}, j_{L+1}) \equiv +\infty
 $$
 \State Sort the items based on $v_i$ so that $v_{\sigma_1^0} \ge v_{\sigma_2^0} \ge \ldots v_{\sigma_N^0}$.
 \State Let $A^0 = G^0 = \{\sigma_1^0, \ldots \sigma_K^0\}$ and $B^0 = \emptyset$. 
 \For{$t = 1, \ldots, L$,}
 \State If $i_t \ne 0$, let $\sigma_t$ be the permutation obtained from transposing $i_t$ and $j_t$ from $\sigma^{t-1}$, and set $B^t = B^{t-1}$.
 \State If $i_t=0$, let $\sigma^t=\sigma^{t-1}$ and $B^t = B^{t-1}\cup \{j_t\}$.
 \State Let $G_t = \{\sigma_1^t, \ldots ,\sigma_K^t\}$ and $A^t = G^t \backslash B^t$.
 \EndFor
 \State Find $A_l$ with the maximum expected reward, i.e., $l^* = \argmax_{l=0,\ldots,L} R(A_l)$. \\
 \Return $A_{l^*}$
 \end{algorithmic}
\end{algorithm}

\subsection{ELSA-UCB with Rank Estimation}

The first extension is the ELSA-UCB algorithm with rank estimation using GIC. We use the samples from random exploration to estimate the rank. We provide theoretical results on rank estimation using GIC in Proposition \ref{prop:rankest_gic}.

\begin{algorithm}[h!]
 \caption{ELSA-GIC : ELSA-UCB with GIC Rank Estimation}
 \label{alg:ELSA-GIC}
 \begin{algorithmic}[1]
 \Require Assortment capacity $K$, learning rate $\eta$, initial parameter estimate $\Phi_0$ and exploration length $T_0$.
 \State Observe item feature vectors $p_i$, $i\in [N]$.
 \For{$t = 1, \ldots, T_0$,} 
 \State Observe the current user feature vector $q_t$.
 \State Randomly select a size $K$ assortment $S_t \in \calS$ and observe user choice $\mathbf{y}_t$.
 \EndFor 
 \State Estimate the rank $\hat{r}$ using GIC (4) using the $T_0$ samples.
 \State Estimate the low-rank matrix $\hat\Phi = \argmax_{rk(\Phi)=\hat{r}} {\calL}_n(\Phi)$ using Algorithm 1.
 \State Estimate the subspace using SVD $\hat\Phi = \hat{U}\hat{D}\hat{V}^\top$.
 \State Initialize $\widehat{\theta}_{t,rtv} \in \bbR^k$ as the truncated vectorization of $\hat{D}$ as in (3).
 \State Rotate the item features $p_i' = \hat{U}^\top p_i$ for $i\in [N]$.
 \For{$t = T_0+1, \ldots, T$,} 
 \State Observe user feature vector $q_t$
 \State Rotate the user feature $q_t' = \hat{V}^\top q_t$.
 \State Compute $x_{it,rtv}$, the truncated vectorization of $p_i'(q_t')^\top$ as in (3).
 \State Compute $z_{it} = x_{it,rtv}^\top \widehat{\theta}_{t,rtv} + \beta_{it}$ as in (5).
 \State Select $S_t = \argmax_{S} \tilde{R}_t(S)$ via StaticMNL and observe choice $\mathbf{y}_t$.
 \State Update the MLE $\widehat{\theta}_{t,rtv}$ by solving (6).
 \EndFor
 \end{algorithmic}
\end{algorithm}

\begin{prop} [Consistency of Rank Estimation] \label{prop:rankest_gic}
    Define $GIC(r) := \calL_n(\hat\Phi_r) - a_n \cdot (d_1 + d_2 - r) \cdot r$, where $a_n = O(1)$ and $na_n \to \infty$ as $n\to \infty$. Let $\hat{r} = \argmin GIC(r)$. Then as $n \to \infty$, $\hat{r} \to r^*$.
\end{prop}

To prove the consistency of the estimated rank, we first identify that the unconstrained MLE is consistent. Then by utilizing the strong convexity and smoothness of the negative log-likelihood function, we prove that the estimation error for the rank-constrained MLE with an over-specified rank is also consistent. Next we show that the estimation error for under-specified rank is lower bounded by a constant that depends on the last singular values of the matrix. Then, we finally prove that the GIC is minimized only at the true rank when the number of samples is sufficiently large. The detailed proof is provided in Appendix \ref{sect:A5}.

\subsection{ELSA-UCB with Batches}

There are scenarios where users arrive in batches, i.e., multiple users arrive at each time point. Under this scenario, the agent cannot update the parameters for every user, but can only update after each batch. We provide the details of the algorithm under this scenario. Note that similar theoretical results hold when batch size is a constant, as the estimation bound holds at each end of the batch.

\begin{algorithm}[h!]
 \caption{ELSA-UCB with Batches}
 \label{alg:ELSA-UCB-Batch}
 \begin{algorithmic}[1]
 \Require Assortment capacity $K$, rank of parameter matrix $r$, learning rate $\eta$, initial parameter estimate $\Phi_0$ and exploration length $B_0$.
 \State Observe item feature vectors $p_i$, $i\in [N]$.
 \For{$b = 1, \ldots, B_0$,} 
 \State Observe the user feature vectors $q_t$, $t=N_{b-1}+1, \ldots, N_b$.
 \State Randomly select a size $K$ assortment $S_t \in \calS$ for each user $t=N_{b-1}+1, \ldots, N_b$.
 \State Observe user choices $\mathbf{y}_t$ for each user $t=N_{b-1}+1, \ldots, N_b$.
 \EndFor 
 \State Estimate the low-rank matrix $\hat\Phi = \argmax_{rk(\Phi)=r} {\calL}_n(\Phi)$ using Algorithm 1.
 \State Estimate the subspace using SVD $\hat\Phi = \hat{U}\hat{D}\hat{V}^\top$.
 \State Initialize $\widehat{\theta}_{t,rtv} \in \bbR^k$ as the truncated vectorization of $\hat{D}$ as in (3).
 \State Rotate the item features $p_i' = \hat{U}^\top p_i$ for $i\in [N]$.
 \For{$b = B_0+1, \ldots, B$,} 
 \State Observe user feature vectors $q_t$ for each user $t=N_{b-1}+1, \ldots, N_b$.
 \State Rotate the user features $q_t' = \hat{V}^\top q_t$ for every user $t=N_{b-1}+1, \ldots, N_b$.
 \State Compute $x_{it,rtv}$, the truncated vectorization of $p_i'(q_t')^\top$ as in (3).
 \State Compute $z_{it} = x_{it,rtv}^\top \widehat{\theta}_{t,rtv} + \beta_{it}$ as in (5).
 \State Select $S_t = \argmax_{S} \tilde{R}_t(S)$ via StaticMNL for each user $t=N_{b-1}+1, \ldots, N_b$.
 \State Observe the batch choice $\mathbf{y}_t$ for each user $t=N_{b-1}+1, \ldots, N_b$.
 \State Update the MLE $\widehat{\theta}_{t,rtv}$ by solving (6).
 \EndFor
 \end{algorithmic}
\end{algorithm}

\section{Future Directions} \label{sect:Appendix_future_work}
First, from a theoretical standpoint, it would be valuable to investigate the precise dependence of the regret bound on the assortment capacity parameter $K$. The current regret bound is not tight on the assortment capacity $K$, and sharper bounds that explicitly characterize this dependence may yield deeper insights. Additionally, it is worth exploring more sophisticated initialization strategies for low-rank models. In particular, recent work by \citet{wang2017unified} suggests that projecting an unconstrained estimator onto a low-rank subspace can significantly reduce sample complexity, a technique that may be applicable to our framework with appropriate adaptation.

Second, although our model accommodates dynamic user and item contexts, our current theoretical analysis assumes a fixed item catalog. In many real-world applications, however, the set of available items evolves over time. Unlike the non-contextual case, where new items require re-learning from scratch, contextual models enable the estimation of utilities for new items using their feature vectors, thereby supporting more efficient adaptation in dynamic environments.

Finally, extending the framework to incorporate pricing decisions offers another compelling direction. Pricing plays a dual role: it influences user preferences (and thus utilities) and directly affects revenue. Several recent works have explored pricing under fixed or linear models \citep{javanmard2020multi,goyal2022dynamic,cai2023doubly,luo2024distribution}, but relatively little is known about joint assortment-pricing strategies in high-dimensional or dual-contextual settings. Investigating regret-minimizing algorithms for joint dynamic pricing and assortment optimization in such environments remains an open and important challenge.

\section{Technical Proofs} \label{sect:AppendixProof}

\allowdisplaybreaks

In this section, we provide technical proofs of the main theoretical statements:
\begin{enumerate}[label=]
    \item \ref{sect:A1} Estimation Bound for Low-Rank Optimization: Proof of Proposition 1.
    \item \ref{sect:A2} Bound on Error Induced by Truncation: Proof of Proposition \ref{prop:lowranksubspace}.
    \item \ref{sect:A3} Regret Upper Bound for the ELSA-UCB Policy: Proof of Theorem 1.
    \item \ref{sect:A4} High Probability Bound for Regret of ELSA-UCB Policy: Proof of Corollary \ref{cor:regretbound_hp}.
    \item \ref{sect:A5} Proof of Rank Consistency using GIC: Proof of Proposition \ref{prop:rankest_gic}.
\end{enumerate}

\subsection{Proof Sketch of Regret Bound} \label{subsect:appendix_proofsketch}

This section provides an outline of the three major steps in the proof of Theorem 1. Detailed proofs are presented in Section \ref{sect:AppendixProof} of the Appendix. Step 1, as shown in Proposition 1, establishes a bound on the bias induced from rotation and truncation of the parameter matrix. Step 2 shows a bound on the difference between the optimistic utility $z_{it}$ defined as (5) and the true utility $v_{it}$ (Lemma \ref{lem:zit_bound}). Finally, using this bound on the optimistic utility, Step 3 derives a bound on the regret of the $t$-th user, and the final result on the bound of cumulative regret (Theorem 1).

As our proposed algorithm ELSA-UCB truncates the lower right sub-block matrix of the rotated matrix, it is sufficient to establish a bound on the sub-block matrix.

\begin{prop}[Bounds for Subspace Estimation]\label{prop:lowranksubspace}
    Suppose assumptions of Proposition 1 hold. Further assume the exploration length $T_0$ satisfies
    $$
    T_0 \ge 4\sigma^2 (C_1\sqrt{d_1d_2} + C_2\sqrt{\log(1/\delta)})^2 + \frac{2B}{\lambda_{\min}(\Sigma)K},
    $$
    for some constants $C_1$, $C_2$ and any $\delta \in (0,1)$, $B>0$. Denote the SVD of the rank-constrained estimator as $\hat{U} \hat{D} \hat{V}^\top$. Then rotated true parameter $\Theta^* = \hat{U}^\top \Phi^* \hat{V}$ satisfies
    \begin{equation} \label{eqn:Sperp}
    \|\Theta^*_{r+1:d_2,r+1:d_1}\|_F^2 \le \frac{288\sigma_1^3}{\sigma_r^3 \cdot \kappa} \cdot \left( \frac{3}{4} + \frac{8T_0}{\kappa B} + \frac{2}{K\|\Sigma\|_2} \right) \frac{Kr\log(\frac{d_1+d_2}{\delta})}{B} S_2^4 =: S_\perp
    \end{equation}
    with probability at least $1-\delta$.
\end{prop}

Proposition \ref{prop:lowranksubspace} allows us to control the gap $S_\perp$ between the reduced subspace and the original subspace by controlling $B$, which relates to the exploration length. Next, we establish a bound on the difference between the optimistic utility $z_{it} = x_{it,rtv}^\top \widehat{\theta}_{t,rtv} + \alpha_t \cdot \|x_{it,rtv}\|_{W_t^{-1}}$ and the true utility $v_{it} = \langle X_{it}, \Theta^* \rangle$ using Proposition \ref{prop:lowranksubspace}. 

\begin{lem}[Bound on Optimistic Utility]\label{lem:zit_bound}
    Let $z_{it} = x_{it,rtv}^\top \widehat{\theta}_{t,rtv} + \beta_{it}$. Then for every $i\in [N]$,
    $$
    0 \le z_{it} - \langle X_{it}, \Theta^* \rangle \le 2 \beta_{it} = 2  \alpha_t \|x_{it,rtv}\|_{W_t^{-1}} + 2S_2^2 \cdot S_\perp
    $$
    with probability at least $1-t^{-2}$.
\end{lem}

Note that $z_{it}$ is greater than the true utility $v_{it} = \langle X_{it}, \Theta^*\rangle$ of item $i$ to user $t$ with probability at least $1-t^{-2}$. Let us denote $R(\{w_{i}\}, S)$ as the expected revenue of offering assortment $S$ under the MNL model assuming utility $w_i$ for each item. Let $S_t^*$ be the optimal assortment with respect to the true revenue $v_{it}$ and $S_t$ be the optimal assortment with respect to the optimistic revenue $z_{it}$. Then from the definition we have the following bound for regret $r_t$ for user $t$:
\begin{equation}
    r_t = R(\{v_{it}\}, S_t^*) - R(\{v_{it}\}, S_t) 
    \le R(\{z_{it}\}, S_t^*) - R(\{v_{it}\}, S_t) 
    \le R(\{z_{it}\}, S_t) - R(\{v_{it}\}, S_t). \label{eqn:rt}
\end{equation}
Moreover, we can bound the last term using the result from \citet{oh2021multinomial} that establishes a bound on the difference of expected reward of an assortment given different utilities:
\begin{lem}[Lemma 5 of \citet{oh2021multinomial}] \label{lem:optimistic_assortment}
    Let $u_i, u_i', i\in [N]$ be utilities and suppose $0\le r_i \le R$ for every item $i\in [N]$. Then $R(\{u_i\}, S) - R(\{u_i'\}, S) \le R \cdot \max_{i\in S} |u_i - u'_i|$.
\end{lem}

Combining (\ref{eqn:rt}), Lemmas \ref{lem:zit_bound} and \ref{lem:optimistic_assortment}, we have
\begin{align*}
{\calR}_T &= \sum_{t=1}^{T_0} r_t + \sum_{t=T_0+1}^T r_t \\
&\le R_{\max} T_0 + R_{\max} \left( \sum_{t=T_0+1}^T 2\alpha_t \max_{i\in S_t} \|x_{it,rtv}\|_{W_t^{-1}} + 2S_2^2 S_\perp(T-T_0) \right).
\end{align*}

Lastly, presenting an upper bound on $\sum_{t=T_0+1}^T 2\alpha_t \max_{i\in S_t} \|x_{it,rtv}\|_{W_t^{-1}}$, we obtain our final result Theorem 1. Detailed proofs of the theoretical results are presented in Section \ref{sect:AppendixProof}.

\subsection{Proof of Proposition 1}
\label{sect:A1}

\subsubsection{Sketch of Proof for Proposition 1}

Proof of Proposition 1 consists of three main steps. The first step is establishing a one-step bound of Algorithm 1 (Lemma \ref{lem:onestep} - \ref{lem:BG}). Next, we establish bounds on the eigenvalues of $V_n$ (Lemma \ref{lem:lamVn}) to verify the constants of Lemma \ref{lem:RSC}. Finally we extend the one-step result to the converged parameter (Lemma \ref{lem:intermediate}) to achieve Proposition 1.

We begin by re-stating the results on convergence of factored gradient descent algorithm on general loss functions from \citet{zhang2023generalized}. Since the result on \citet{zhang2023generalized} is a general result for where the loss is a joint function of a low-rank matrix and a sparse tensor, we state the simplified result without the assumptions on the tensor portion. First we define the minimal Frobenius norm under rotation.

\begin{defn}
    For $Z_1, Z_2 \in \bbR^{d\times r}$, $d(Z_1,Z_2)$ is defined as the minimal Frobenius norm between $Z_1$ and $Z_2$ under rotation i.e.,
    $$
    d(Z_1, Z_2) = \min_{R \in \mathbb{Q}_r} \|Z_1 - Z_2R\|_F ,
    $$
    where $\mathbb{Q}_r$ is the space of $r$-dimensional rotation matrices, i.e., $\mathbb{Q}_r = \{ R \in \bbR^{r\times r}: R^\top R = R R^\top = I_r\}$.
\end{defn}

Next we use the following Lemma to illustrate the one-step rate decrease in terms of the new metric.

\begin{lem}[Lemma 1 from \citep{zhang2023generalized}]\label{lem:onestep}
    Suppose $\ell$ satisfies the restricted convexity and restricted smoothness i.e., for any matrices $\Phi_1, \Phi_2 \in \mathbb{B}(\Phi^*, \kappa_1)$ with rank at most $r$,
    $$
    \frac{\mu}{2} \|\Phi_2 - \Phi_1\|_F^2 \le \ell(\Phi_2) - \ell(\Phi_1) - \langle \nabla \ell(\Phi_1), \Phi_2 - \Phi_1 \rangle \le \frac{L}{2} \|\Phi_2 - \Phi_1\|_F^2 .
    $$
    Let $\sigma_1, \ldots, \sigma_r$ denote the singular values of $\Phi^*$. Let $c_1, c_2$ be a constant such that $c_1 \le \min\{1/32, \mu/(192L^2)$, $c_2 \le \sqrt{\min\{\mu,2\}(6L+4)}$ and consider the step size $\delta = c_1/\sigma_1$. Let $Z^* = [U^* ; V^*]$, and $Z^{(t)} = [U^{(t)}, V^{(t)}$. If $d(Z^{(t)}, Z^*) \le c_2 \sqrt{\sigma_r}$, then Algorithm 1 satisfies
    $$
    d^2(Z^{(t+1)}, Z^*) \le \rho d^2(Z^{(t)}, Z^*) - \frac{\delta \mu}{4} \|\Phi^{(t)} - \Phi^*\|_F^2 + C_1 \|\nabla \ell(\Phi^*)\|_2^2 ,
    $$
    where $\rho = 1- \delta\mu \sigma_r/16$, $C_1 = 48r\delta^2 \sigma_1 + 2\delta(8r/\mu + r/L)$.
\end{lem}

Next step is checking the assumptions of Lemma \ref{lem:onestep}. First we begin with defining $V_n$, which explains the variability of the features:
$$
V_n := \sum_{t=1}^n \sum_{i\in S_t} x_{it}x_{it}^\top .
$$

Then we have the following result on the restricted strong convexity and smoothness of the loss function:

\begin{lem} [Restricted Strong Convexity and Smoothness] \label{lem:RSC}
$$
\frac{\kappa}{n} \lambda_{\min}(V_n) \|Y- X\|_F^2 \le {\calL}_n(Y) - {\calL}_n(X) - \langle \nabla {\calL}_n (X), Y-X \rangle \le \frac{1}{n} \lambda_{\max}(V_n) \|Y -X\|_F^2 ,
$$
as long as we can guarantee $x,y \in \calB(\phi^*, 1)$, or equivalently $X,Y \in \calB_F(\Phi^*, 1)$. 
\end{lem}

Next to utilize the result of Lemma \ref{lem:onestep}, we need a bound on the last term of the inequality, the gradient of the loss function evaluated in the true parameter. 

\begin{lem}[Bounded Gradient] \label{lem:BG}
    $$
    \left\| \frac{1}{n} \sum_{t=1}^n \sum_{i\in S_t} (y_{it} - p_t(i|\Phi^*)) p_i q_t^\top \right\|_2 \le \sqrt{\frac{K \log(\frac{d_1+d_2}{\delta})}{n}} S_2^2 =: \epsilon(n,\delta) ,
    $$
    with probability at least $1-\delta$. 
\end{lem}

Plugging in the results of Lemma \ref{lem:RSC} and \ref{lem:BG} to Lemma \ref{lem:onestep}, we have the following result on the converged estimator $\hat\Phi$:
\begin{lem} \label{lem:intermediate}
    If $\|\Phi_0 - \Phi^*\| \le c_2 \sqrt{\sigma_r}$ where $c_2 \le \sqrt{\min\{\mu, 2\} \cdot (6L + 4)}$, then:
    \begin{align*}
    \|\hat\Phi - \Phi^*\|_F^2 &\le \frac{288\sigma_1^2}{\sigma_r\cdot \kappa} \cdot \left( \frac{3}{4} + \frac{8n}{\kappa \lambda_{\min}(V_n)} + \frac{n}{\lambda_{\max}(V_n)} \right) \frac{Kr\log(\frac{d_1+d_2}{\delta})}{\lambda_{\min}(V_n)} S_2^4.
    \end{align*}
\end{lem}

Now it remains to establish a bound on $\lambda_{\min}(V_n)$ and $\lambda_{\max}(V_n) = \|V_n\|_2$.

\begin{lem}\label{lem:lamVn}
    If
    $$
    n \ge 4\sigma^2 (C_1\sqrt{d_1d_2} + C_2\sqrt{\log(1/\delta)})^2 + \frac{2B}{\lambda_{\min}(\Sigma)K} ,
    $$
    for some fixed constant $C_1, C_2$, with probability at least $1-\delta$ we have $\lambda_{\min}(V_n) \ge B$ and
    $$
    \frac{\|V_n\|_2}{n} \ge \frac{1}{2} K\|\Sigma\|_2 .
    $$
\end{lem}

Combining Lemma \ref{lem:intermediate} and Lemma \ref{lem:lamVn}, we have Proposition 1. In the remaining part of the section, we provide proof of the Lemmas.

\subsubsection{Proof of Lemma \ref{lem:RSC}}

For proof of Lemma \ref{lem:RSC}, we use Taylor expansion and bounds on the second moment of the loss function. From the definition of the loss function, We can easily check that
$$
\nabla^2 {\calL}_n (\phi) = \frac{1}{n} \sum_{t=1}^n \left[ \sum_{i\in S_t} p_t(i|\phi) x_{it} x_{it}^\top - \sum_{i\in S_t} \sum_{j\in S_t} p_t(i|\phi) p_t(j|\phi) x_{it} x_{jt}^\top \right] .
$$

Also note that from Taylor expansion we have:
\begin{equation} \label{eqn:hessian}
{\calL}_n(Y) - {\calL}_n(X) - \langle \nabla {\calL}_n (X), Y-X \rangle = (y-x)^\top \nabla^2 {\calL}_n(\bar{\phi}) (y-x) ,
\end{equation}

where $x = \vect(X), y = \vect(Y)$, $\bar{\phi} = cx + (1-c)y$ for some $c\in (0,1)$. Also note that for any vector $a,b$, $(a-b)(a-b)^\top \succeq 0$ and therefore:
\begin{equation}\label{eqn:vectorineq}
     aa^\top + bb^\top \succeq ab^\top + ba^\top.
\end{equation}

Using this fact we have:
\begin{align*}
     & \sum_{i\in S_t} p_t(i|\phi) x_{it} x_{it}^\top - \sum_{i\in S_t} \sum_{j\in S_t} p_t(i|\phi) p_t(j|\phi) x_{it} x_{jt}^\top \\
     =& \sum_{i\in S_t} p_t(i|\phi) x_{it} x_{it}^\top - \frac{1}{2} \sum_{i\in S_t} \sum_{j\in S_t} p_t(i|\phi) p_t(j|\phi) (x_{it} x_{jt}^\top + x_{jt} x_{it}^\top) \\
     \succeq & \sum_{i\in S_t} p_t(i|\phi) x_{it} x_{it}^\top - \frac{1}{2} \sum_{i\in S_t} \sum_{j\in S_t} p_t(i|\phi) p_t(j|\phi) (x_{it} x_{it}^\top + x_{jt} x_{jt}^\top) &(\because (\ref{eqn:vectorineq}))\\
     = & \sum_{i\in S_t} p_t(i|\phi) x_{it} x_{it}^\top - \sum_{i\in S_t} \sum_{j\in S_t} p_t(i|\phi) p_t(j|\phi) x_{it} x_{it}^\top \\
     = & \sum_{i\in S_t} p_t(i|\phi) p_t(0|\phi) x_{it}x_{it}^\top . & \left(\because 1 - \sum_{j\in S_t} p_t(j|\phi) = p_t(0|\phi) \right)
\end{align*}

Now assume $\bar\phi \in \calB(\phi^*, 1)$, and suppose $p_t(i|\phi) p_t(0|\phi) \ge \kappa$ for every $t$ and $\phi \in \calB(\phi^*, 1)$. Applying the inequality above to (\ref{eqn:hessian}), we have:
\begin{align*}
    &(y-x)^\top \nabla^2 {\calL}_n(\bar{\phi}) (y-x) \\
    \ge& (y-x)^\top \left( \frac{1}{n}  \sum_{t=1}^n \sum_{i\in S_t} p_t(i|\bar{\phi}) p_t(0|\bar\phi) x_{it} x_{it}^\top \right) (y-x) \\
    \ge& (y-x)^\top \frac{\kappa}{n} V_n (y-x) &\left( \because V_n = \sum_{t=1}^n\sum_{i\in S_t} x_{it} x_{it}^\top \right) \\
    \ge& \frac{\kappa}{n} \lambda_{\min}(V_n) \|y-x\|^2 \\
    =& \frac{\kappa}{n} \lambda_{\min}(V_n) \|Y-X\|_F^2 , & (\because \vect(X) = x, \vect(Y) = y)
\end{align*}

which gives the left side of the inequality. Now also consider that:
\begin{align*}
    & \sum_{i\in S_t} p_t(i|\phi) x_{it} x_{it}^\top - \sum_{i\in S_t} \sum_{j\in S_t} p_t(i|\phi) p_t(j|\phi) x_{it} x_{jt}^\top \\
    =& \sum_{i\in S_t} p_t(i|\phi) x_{it} x_{it}^\top - \left(\sum_{i\in S_t} p_t(i|\phi) x_{it}\right) \left(\sum_{i\in S_t} p_t(i|\phi) x_{it}\right)^\top  \\
    \preceq& \sum_{i\in S_t} p_t(i|\phi) x_{it} x_{it}^\top \\
    \preceq & \sum_{i\in S_t} x_{it} x_{it}^\top . &(\because p_t(i|\phi) \le 1)
\end{align*}

Applying the inequality to (\ref{eqn:hessian}), we have:
\begin{align*}
    (y-x)^\top \nabla^2 {\calL}_n(\bar{\phi}) (y-x) &\le (y-x)^\top \left( \frac{1}{n}  \sum_{t=1}^n \sum_{i\in S_t} x_{it} x_{it}^\top \right) (y-x) \\
    &= \frac{1}{n} (y-x)^\top V_n (y-x) \\
    &\le \frac{1}{n} \lambda_{\max}(V_n) \|Y-X\|_F^2 ,
\end{align*}

which completes the proof.

\subsubsection{Proof of Lemma \ref{lem:BG}}

We use the following Bernstein inequality for rectangular matrices from \citep{tropp2012user}:
\begin{lem}[(Bernstein Inequality for Rectangular Matrices, Theorem 1.6 of \citep{tropp2012user}] \label{lem:matrixbernstein}
    Suppose $Z_k \in \bbR^{d_1\times d_2}$ are independent random matrices. Suppose $\bbE Z_k = 0$ and $\|Z_k\|_2 \le R$ almost surely. Define
    $$
    \sigma^2 := \max \left\{ \left\| \sum_k \bbE (Z_k Z_k^\top) \right\|, \left\|\sum_k \bbE(Z_k^\top Z_k) \right\| \right\} .
    $$
    Then for all $t\ge 0$,
    $$
    \bbP\left\{ \left\| \sum_k Z_k \right\| \ge t \right\} \le (d_1 + d_2) \cdot \exp \left( \frac{-t^2/2}{\sigma^2 + Rt/3} \right) .
    $$
\end{lem}

Let us plug in $Z_t = \sum_{i\in S_t} (y_{it} - p_t(i|\Phi^*)) p_i q_t^\top$ to apply Lemma $\ref{lem:matrixbernstein}$. For brevity, let us denote $\epsilon_{it} = y_{it} - p_t(i|\Phi^*)$. We can easily check that $\bbE \epsilon_{it} = 0$.

Now to verify $R$ in the conditions of Lemma \ref{lem:matrixbernstein}, note that $Z_t = \sum_{i\in S_t} \epsilon_{it} p_i q_t^\top$, therefore $\bbE Z_t = 0$. Now note that $\|p_i\|, \|q_t\| \le S_2$ for every $i \in [N]$ and $t\in [T]$. Note that 
$$
\|p_iq_t^\top v\| = q_t^\top v \|p_i\| \le \|q_t\| \cdot \|v\| \cdot \|p_i\| \le S_2^2 \|v\|.
$$

Therefore $\|p_iq_t^\top\|_2 \le S_2^2$. Now let $i(t)$ be the item selected for user $t$, where $i(t) = 0$ if no items were selected. Also let $\epsilon_{0t} = 0$. Note that $1 \ge \epsilon_{i(t)t} \ge 0$ and $\epsilon_{jt} \le 0$ for $j\ne i(t)$ and moreover $0 \ge \sum_{S(t)\backslash \{i(t)\}} \epsilon_{jt} \ge -1$. Therefore,
\begin{align*}
    \|Z_t\|_2 &= \left\| \sum_{i\in S_t} \epsilon_{it} \cdot p_i q_t^\top \right\|_2 \\
    &\le \epsilon_{i(t),t} \cdot  \| p_{i(t)} q_t^\top \|_2 + \sum_{S(t)\backslash \{i(t)\}} (-\epsilon_{jt}) \cdot \| p_i q_t^\top \|_2 \\
    &\le \epsilon_{i(t),t} S_2^2 + \sum_{S(t)\backslash \{i(t)\}} (-\epsilon_{jt}) \cdot S_2^2 \\
    &\le 2S_2^2,
\end{align*}

which allows us to write $R = 2S_2^2$. Next we verify $\sigma^2$. Note that:
\begin{align*}
    \sum_{t=1}^n \bbE Z_t Z_t^\top &=  \bbE \left[ \sum_{t=1}^n \sum_{i\in S_t} \epsilon_{it}^2 \cdot p_i q_t^\top q_t p_i^\top \right] \\
    &= \sum_{t=1}^n \sum_{i\in S_t} \bbE \epsilon_{it}^2 p_i q_t^\top q_t p_i^\top \\
    &= \sum_{t=1}^n \sum_{i\in S_t} p_t(i|\Phi^*) (1-p_t(i|\Phi^*) p_iq_t^\top q_t p_i^\top ,
\end{align*}
since $\bbE \epsilon_{it}^2 = \Var(\epsilon_{it}) = \Var(y_{it}) = (\bbE y_{it})(1- \bbE y_{it})$. Also note that $\|pp^\top v\| = |p^\top v| \cdot \|p\| \le \|p\|^2 \cdot \|v\|$, so $\|pp^\top\|_2 = \|p\|^2$. Therefore,
\begin{align*}
    \left\| \sum_{t=1}^n \bbE Z_t Z_t^\top \right\|_2 &= \left\| \sum_{t=1}^n \sum_{i\in S_t} p_t(i|\Phi^*) (1-p_t(i|\Phi^*) (q_t^\top q_t) p_i p_i^\top \right\|_2 \\
    &\le \sum_{t=1}^n \sum_{i\in S_t} p_t(i|\Phi^*) (1-p_t(i|\Phi^*)) \|q_t\|_2^2 \|p_ip_i^\top \|_2 \\
    &\le \frac{1}{4} \sum_{t=1}^n \sum_{i\in S_t} \|p_i\|_2^2 \|q_t\|_2^2 &(\because p(1-p) \le 1/4)\\
    &\le \frac{nK}{4} S_2^4. &(\because \|p_i\|_2, \|q_t\|_2 \le S_2)
\end{align*}

and similarly we have:
$$
\left\| \sum_{t=1}^n \bbE Z_t^\top Z_t \right\|_2 \le \frac{nK}{4}S_2^4 .
$$

Now applying the Lemma \ref{lem:matrixbernstein} with $R= 2S_2^2$ and $\sigma^2 = \frac{nK}{4}S_2^4$ we have:
$$
\bbP \left( \left\|\frac{1}{n} \sum_{t=1}^n Z_t \right\|_2 \ge t \right) \le (d_1 + d_2) \exp \left(-\frac{6nt^2}{S_2^2(3KS_2^2 + 8t)} \right) .
$$

Let $\eta = \log\frac{d_1+d_2}{\delta}$ for some $\delta \in (0,1)$. Assuming $n \ge \frac{64}{9} K \eta^2$, we have $\frac{3}{8} K S_2^2 \ge \sqrt{\frac{K\eta}{n}} S_2^2$. Now select $t = \sqrt{\frac{K\eta}{n}} S_2^2 \le \frac{3}{8}KS_2^2$. Then we have $8t \le 3KS_2^2$ and $nt^2 \ge K\eta S_2^4$. Therefore we have
\begin{align*}
\bbP \left( \left\|\frac{1}{n} \sum_{t=1}^n Z_t \right\|_2 \ge t \right) &\le (d_1 + d_2) \exp \left(-\frac{6nt^2}{S_2^2(3KS_2^2 + 8t)} \right) \\ 
&\le (d_1 + d_2) \exp(-\eta) = \delta .
\end{align*}
Therefore,
$$
\left\| \frac{1}{n} \sum_{t=1}^n \sum_{i\in S_t} (y_{it} - p_t(i|\Phi^*)) p_i q_t^\top \right\|_2 \le \sqrt{\frac{K \log(\frac{d_1+d_2}{\delta})}{n}} S_2^2 = \epsilon(n,\delta) ,
$$
with probability at least $1-\delta$. This completes the proof of Lemma \ref{lem:BG}.

\subsubsection{Proof of Lemma \ref{lem:intermediate}}

Before moving on to the main proof of Lemma \ref{lem:intermediate}, we first state the modified result from \citep{wang2017unified} and Lemma \ref{lem:onestep}.

\begin{prop}\label{prop:dist}
    For $i =1,2$, let $U_i \in \bbR^{d_1\times r}$, $V_i \in \bbR^{d_2\times r}$ and $X_i = U_i V_i^\top$. Now denote
    $$
    Z_1 = \begin{pmatrix} U_1\\ V_1 \end{pmatrix}, \quad 
    Z_2 = \begin{pmatrix} U_2 \\ V_2 \end{pmatrix} .
    $$
    Then,
    $$
    \|X_1 - X_2\|_F^2 \le 2 (\|Z_2\|_2 + d(Z_1, Z_2))^2 d^2(Z_1, Z_2) .
    $$
\end{prop}

\begin{proof}
    Note that for two positive definite matrices $A,B$ with eigenvalue decomposition $A = P\Lambda P^\top$ and $B = QDQ^\top$, 
    \begin{align*}
        \tr(AB) &= \tr(P \Lambda P^\top QDQ^\top) \\
        &= \tr(Q^\top P \Lambda P^\top Q D) &(\because \tr(AB) = \tr(BA))\\
        &\le \tr(Q^\top P \Lambda P^\top Q) \cdot \|D\|_2 &(\because \tr(AB) \le \tr(A)\cdot \|B\|_2)\\
        &= \tr(\Lambda) \cdot \|D\|_2 \\
        &= \tr(A) \cdot \|B\|_2 .
    \end{align*}

    This implies:
    \begin{align*}
        \|AB\|_F^2 &= \tr(ABB^\top A^\top) \\
        &= \tr(BB^\top A^\top A) \\
        &\le \tr(BB^\top) \|A^\top A\|_2 \\
        &= \|B\|_F^2 \|A\|_2^2
    \end{align*}
    
    and therefore $\|AB\|_F \le \|B\|_F \|A\|_2$. Now consider $R \in \mathbb{Q}_r$. Note that:
    \begin{align*}
    \|U_1(V_1- V_2R)^\top \|_F^2 &= \tr(U_1 (V_1 - V_2R)^\top (V_1 - V_2R) U_1^\top) \\
    &= \tr((V_1-V_2R)^\top (V_1 - V_2R) U_1^\top U_1) \\
    &\le \|U_1^\top U_1\|_2 \cdot \tr((V_1 - V_2R)^\top (V_1-V_2R)) \\
    &=  \|U_1\|_2^2 \|V_1 - V_2R\|_F^2 .
    \end{align*}
    Similarly we have $\|(U_1 - U_2R) R^\top V_2^\top\|_F^2 \le \|V_2\|_2^2 \|U_1 - U_2R\|_F^2$. Now note that for arbitrary vector $v$,
    $$
    \|Z_iv\|_2^2 = \|U_iv\|_2^2 + \|V_iv\|_2^2 \ge \|U_iv\|_2^2 .
    $$
    
    So we have $\|Z_i\|_2 \ge \|U_i\|_2$ and similarly $\|Z_i\|_2 \ge \|V_i\|_2$. Now note that:
    \begin{align*}
    \|Z_1\|_2 &\le \|Z_1 - Z_2R\|_2 + \|Z_2R\|_2 \\
    &= \|Z_1 - Z_2R \|_2 + \|Z_2\|_2 &(\because R \in \mathbb{Q}_r)\\
    &\le \|Z_1 - Z_2R\|_F + \|Z_2\|_2 . &(\because \|A\|_2 \le \|A\|_F)
    \end{align*}

    And this holds for arbitrary $R \in \mathbb{Q}_r$. Therefore, 
    \begin{equation} \label{eqn:rotineq}
    \|Z_1\|_2 \le \|Z_2\|_2 + d(Z_1, Z_2) .
    \end{equation}
    
   And as $(a+b)^2\le a^2+b^2$, we have $\|Z_1\|_2^2 \le 2(\|Z_2\|_2^2 + d(Z_1, Z_2))$. Therefore we have:
    \begin{align*}
        \|X_1 - X_2\|_F^2 &= \|U_1 V_1^\top - U_2 V_2^\top \|_F^2 \\
        &= \|U_1 V_1^\top - U_1 R^\top V_2^\top + U_1 R^\top V_2^\top - U_2 R R^\top V_2^\top \|_F^2 \\
        &\le (\|U_1(V_1 - V_2R)^\top \|_F+ \|(U_1 - U_2R) R^\top V_2^\top \|_F)^2 \\
        &\le 2(\|U_1(V_1 - V_2R)^\top \|_F^2 + \|(U_1 - U_2R)R^\top V_2^\top\|_F^2) &(\because (a+b)^2 \le 2(a^2+b^2))  \\
        &\le 2 (\|V_2\|_2^2\|U_1 - U_2R\|_F^2 + \|U_1\|_2^2\|V_1 - V_2R\|_F^2) &(\because \|AB\|_F^2 \le \|A\|_F\|B\|_2)\\
        &\le 2 (\|Z_2\|_2^2\|U_1 - U_2R\|_F^2 + \|Z_1\|_2^2\|V_1 - V_2R\|_F^2) \\
        &\le 2 (\|Z_2\|_2 + d(Z_1, Z_2))^2 (\|U_1 - U_2R\|_F^2 + \|V_1 - V_2R\|_F^2) &(\because (\ref{eqn:rotineq})\\
        &= 2 (\|Z_2\|_2 + d(Z_1, Z_2))^2  \|Z_1 - Z_2R\|_F^2 .
    \end{align*}
    Since this holds for arbitrary $R \in \mathbb{Q}_r$, The statement holds. 
\end{proof}

Now we can finally prove Lemma \ref{lem:intermediate}.

\begin{proof}
    Assume $d(Z^{(0)}, Z^*) \le \kappa_1$ and $C_1 \epsilon(n,\delta) \le (1-\rho) \kappa_1$. Then by induction and using \ref{lem:onestep}, we can check that $d(Z^{(t)}, Z^*) \le \kappa_1$ for every $t\ge 0$. Also we have:
    \begin{align*}
    d^2(Z^{(t)}, Z^*) &\le \rho^t d^2(Z^{(0)}, Z^*) + \frac{C_1(1-\rho^t)}{1-\rho} \epsilon^2(n,\delta) \\
    &\le \rho^t \kappa_1^2 + \frac{C_1}{1-\rho} \epsilon^2(n,\delta) .
    \end{align*}

    Note that $\|Z^*\|_2 \le 2\|\Phi^*\|_2 = 2\sigma_1$. Using Proposition \ref{prop:dist}, we have:
    \begin{align*}
        \|\Phi^{(t)} - \Phi^*\|_F^2 &\le 2(2\sigma_1 + d(Z^{(t)}, Z^*))^2 d^2(Z^{(t)}, Z^*)\\
        &\le (8\sigma_1^2 + 2 d^2(Z^{(t)}, Z^*)) d^2(Z^{(t)}, Z^*) . &(\because (a+b)^2 \le 2(a^2+b^2))
    \end{align*}
    
    As $t\to \infty$, we have:
    $$
    \|\hat{\Phi} - \Phi^*\|_F^2 \le \left( 8\sigma_1^2 + \frac{2C_1}{1-\rho} \epsilon^2(n,\delta) \right) \cdot \frac{C_1}{1-\rho} \epsilon^2(n,\delta) .
    $$

    Now assume $\epsilon^2(n,\delta) \le \frac{(1-\rho)\sigma_1^2}{2C_1}$. Then we have:
    $$
    \|\hat\Phi - \Phi^*\|_F^2 \le \frac{9\sigma_1^2 C_1}{1-\rho} \epsilon^2(n,\delta) .
    $$

    Now note that $C_1 = r(48\delta^2\sigma_1 + 2\delta(8/\mu + 1/L))$. So we have:
    $$
    \|\hat\Phi - \Phi^*\|_F^2 \le \frac{9\sigma_1^2}{1-\rho}\left(48\delta^2 \sigma_1 + 2\delta\left( \frac{8}{\mu} + \frac{1}{L}\right) \right) r \epsilon^2(n,\delta) .
    $$

    Plugging in the values of $\delta = c/\sigma_1$ and $\rho = 1 - \delta\mu\sigma_4/16$ we have:
    $$
    \|\hat\Phi - \Phi^*\|_F^2 \le 9\sigma_1^2 \cdot \frac{32}{\mu\sigma_r} \left( 24c_1 + \frac{8}{\mu} + \frac{1}{L} \right) r \epsilon^2(n,\delta) .
    $$

    Now note that the constants are defined as:
    \begin{align*}
        \mu &= \frac{\kappa}{n} \lambda_{\min}(V_n) \\
        L &= \frac{1}{n} \lambda_{\max}(V_n) \\
        \epsilon(n,\delta) &= \sqrt{\frac{K\log(\frac{d_1+d_2}{\delta})}{n}} S_2^2 \\
        c_1 &\le \min \left( \frac{1}{32}, \frac{\mu}{192L^2} \right) \le \frac{1}{32} .
    \end{align*}

    Plugging in the values we achieve:
    \begin{align*}
    \|\hat\Phi - \Phi^*\|_F^2 &\le \frac{288\sigma_1^2}{\sigma_r} \cdot \frac{n}{\kappa \lambda_{\min}(V_n)} \cdot \left( 24 c_1 + \frac{8n}{\kappa \lambda_{\min}(V_n)} + \frac{n}{\lambda_{\max}(V_n)} \right) r \cdot \frac{K\log(\frac{d_1+d_2}{\delta})}{n} S_2^4 \\
    &= \frac{288\sigma_1^2}{\sigma_r\cdot \kappa} \cdot \left( \frac{3}{4} + \frac{8n}{\kappa \lambda_{\min}(V_n)} + \frac{n}{\lambda_{\max}(V_n)} \right) \frac{Kr\log(\frac{d_1+d_2}{\delta})}{\lambda_{\min}(V_n)} S_2^4 ,
    \end{align*}
    which completes the proof.

\end{proof}

\subsubsection{Proof of Lemma \ref{lem:lamVn}}

For proof of Lemma \ref{lem:lamVn}, we modify results from \citet{vershynin2010introduction} about concentration inequalities of random matrices with sub-gaussian rows. 

Note that from Assumption 1, $q_t$ are i.i.d. where $\bbE [q_t q_t^\top] = \Sigma_q$. Also note that $S_t$ are uniformly sampled from $[N]$ with size $K$ in the exploration stage, where $S_t$'s and $q_t$'s are independent. As $S_t$ are sampled uniformly, we have $\bbE \sum_{i\in S_t} p_i p_i^\top = \frac{K}{N} \sum_{i=1}^N p_i p_i^\top =: K \Sigma_p$. Now note that:
$$
\bbE \left[ \sum_{t=1}^n \sum_{i\in S_t} x_{it} x_{it}^\top \right] = nK \Sigma_q \otimes \Sigma_p =: nK \Sigma .
$$

Define $a_{it} = \Sigma^{-1/2} x_{it}$. Consider the matrix $A \in \bbR^{nK\times d_1d_2}$ with $a_{it}$ as its rows. Note that $A^\top A = \sum_{t}\sum_i a_{it} a_{it}^\top = \Sigma^{-1/2} V_n \Sigma^{-1/2}$. Instead of analyzing $V_n$ directly, we first claim the following concentration inequality:
$$
\left\| \frac{1}{nK} A^\top A - I \right\| \le \max(\delta, \delta^2) =: \epsilon .
$$

We use the following lemma from \citet{vershynin2010introduction}.

\begin{lem}[Lemma 5.4 of \citet{vershynin2010introduction}]\label{netlem}
    Let $A$ be a symmetric $n\times n$ matrix and let $\calN_\epsilon$ be an $\epsilon$-net of $S^{n-1}$ for some $\epsilon \in [0,1)$. Then, 
    $$
    \|A\| = \sup_{x\in S^{n-1}} |\langle Ax, x \rangle| \le (1-2\epsilon)^{-1} \sup_{x\in \calN_\epsilon} |\langle Ax, x \rangle| .
    $$
\end{lem}

Applying this lemma, it suffices to show
$$
\max_{v\in \calN} \left| \frac{1}{NK} \|Av\|_2^2 - 1 \right| \le \frac{\epsilon}{2} ,
$$
where $\calN$ is the $\frac{1}{4}$-net of sphere $S^{d_1d_2-1}$. Now note that
\begin{align*}
    \|Av\|_2^2 &= \sum_{t=1}^n \sum_{i\in S_t} (a_{it}^\top v)^2 \\
    &= \sum_{t=1}^n \sum_{i\in S_t} \left(x_{it} \Sigma^{-1/2} v \right)^2 &(\because \text{Definition of }a_{it})\\
    &= \sum_{t=1}^n \sum_{i\in S_t} \langle p_i q_t^\top, U \rangle^2 \\
    &= \sum_{t=1}^n \sum_{i\in S_t} (p_i^\top U q_t)^2 ,
\end{align*}
where $U\in \bbR^{d_2\times d_1}$ with $\vect(U) = \Sigma^{-1/2}v$. Now note that:
$$
\vect(U) = \Sigma^{-1/2} v = (\Sigma_q^{-1/2} \otimes \Sigma_p^{-1/2}) v = \vect( \Sigma_p^{-1/2} V \Sigma_q^{-1/2}) ,
$$
where $\vect(V) = v$. Therefore:
$$
p_i^\top U q_t = \left( \Sigma_p^{-1/2} p_i \right)^\top V \left(\Sigma_q^{-1/2} q_t \right) = x_i^\top V y_t .
$$

To show this we first show that these sums are actually sub-exponential random variables. Suppose $y_t$'s are sub-Gaussian random vector with $\|y_t\|_{\psi_2} \le \sigma$. For $p\ge 1$ and $v\in S^{d_1-1}$, $\| v^\top y_t \|_{\psi_2} \le \|y_t\|_{\psi_2} \le \sigma$. Note that $2p \ge 1$. By definition of $\|\cdot\|_{\psi_2}$, we have $(\bbE \|v^\top y_t\|^{2p})^{1/2p}/{\sqrt{2p}} \le \sigma$ which gives
\begin{equation} \label{eqn:subgaussian}
    \frac{(\bbE\|v^\top y_t\|^{2p})^{1/p}}{p} \le 2\sigma^2
\end{equation}

Also fix $S\subset [N]$ as a set of size $K$ independent of $u$. Moreover assume that $\max_{i\in [N]} \|x_i\| \le S_v < \infty$ and note that $\|v\|_2^2 = \|V\|_F^2 =1$ as $v\in \calN$. Then,
\begin{align*}
    \sup_{p\ge 1} \frac{\left( \bbE \left[ \sum_{i\in S} (x_i^\top V y_t)^2 \right]^p \right)^{1/p}}{p} 
    &\le \sup_{p \ge 1} \frac{ \left( \bbE \frac{K^p}{K} \sum_{i\in S} (x_i^\top V y_t)^{2p} \right)^{1/p}}{p} &(\because \text{Jensen's Inequality}) \\
    &\le K \sup_{p\ge 1} \frac{\left( \sum_{i\in S} \bbE[x_i^\top V y_t]^{2p} \right)^{1/p}}{p} &(\because K\ge 1)\\
    &\le K \sup_{p\ge 1} \frac{\sum_{i\in S}(\bbE[x_i^\top Vy_t]^{2p})^{1/p}}{p} \\
    &\le K \sum_{i\in S} \sup_{p\ge 1} \frac{(\bbE[x_i^\top Vy_t]^{2p})^{1/p}}{p} \\
    &\le K \sup_{p\ge 1} \left( \bbE \left[ \sum_{i \in S} 2\sigma^2 \cdot \|V^\top x_i\|_2^{2} \right]^p \right)^{1/p} &(\because (\ref{eqn:subgaussian}))\\
    &= 2\sigma^2 K \sup_{p\ge 1} \left(\bbE \left[ \sum_{i\in S} \|V^\top x_i\|_2^{2} \right]^p \right)^{1/p} \\
    &\le 2\sigma^2 K \cdot \max_S \sum_{i\in S} \|V^\top x_i \|_2^{2} \\
    &\le 2\sigma^2 K^2 S_v^2 \|V\|_2^2 \\
    &\le 2\sigma^2 K^2 S_v^2 \|V\|_F^2 \\
    &\le 2\sigma^2 K^2 S_v^2 .
\end{align*}

Therefore $Z_t = \sum_{i\in S} (x_i^\top V y_t)^2$ is sub-exponential with parameter $2\sigma^2 K^2 S_v^2$. Therefore, from Remark 5.18 from \citet{vershynin2010introduction}, centered $Z_t$, i.e., $Z_t - K$ are independent centered sub-exponential variables with parameter $4\sigma^2 K^2 S_v^2$. Now consider the following result for centered sub-exponential random variables:
\begin{lem}[Corollary 5.17 of \citep{vershynin2010introduction}]\label{subexpcon}
    Let $X_1, \ldots, X_N$ be independent centered sub-exponential random variables, and let $K = \max_i \|X_i\|_{\psi_1}$. Then, for every $\epsilon \ge 0$, we have
    $$
    \bbP \left\{ \left| \sum_{i=1}^N X_i \right| \ge \epsilon N \right\} \le 2 \exp \left[ -c\min \left( \frac{\epsilon^2}{K^2}, \frac{\epsilon}{K} \right) N \right] .
    $$
\end{lem}

Using exponential deviation inequality, and assuming $8\sigma^2 S_v^2 \ge 1$, we have:
\begin{align*}
\bbP \left( \left| \frac{1}{nK} \|Av\|_2^2 - 1  \right| \ge \frac{\epsilon}{2} \right) &= \bbP \left( \left| \frac{1}{n} \sum_{i=1}^n (Z_t - K) \right| \ge \frac{K\epsilon}{2} \right) \\
&\le 2\exp \left[ - c \left( \left(\frac{K\epsilon}{8\sigma^2 K^2 S_v^2} \right)^2  \wedge \left(\frac{K\epsilon}{8\sigma^2 K^2 S_v^2} \right) \right) n \right] \\
&\le 2\exp \left[ - \frac{c}{64\sigma^2 K^2 S_v^4} (\epsilon^2 \wedge 8\sigma^2 S_v^2 \epsilon) n \right] \\
&\le 2 \exp \left[ -\frac{c}{64\sigma^2 K^2 S_v^4} (\epsilon^2 \wedge \epsilon) n \right] \\
&= 2\exp \left[ - \frac{c}{64\sigma^2 K^2 S_v^4} \delta^2 n \right] \\
& \le 2 \exp \left[ - \frac{c}{64\sigma^2 K^2 S_v^4} (C^2 d_1d_2 + t^2) \right] ,
\end{align*}
where $\epsilon = \max(\delta, \delta^2)$ and $\delta = C\sqrt{\frac{d_1d_2}{n}} + \frac{t}{\sqrt{n}}$. Now consider the following lemma about covering number of the sphere. 

\begin{lem}[Lemma 5.2 of \citep{vershynin2010introduction}]\label{coversphere}
    The unit Euclidean sphere $S^{n-1}$ equipped with the Euclidean metric satisfies for every $\epsilon < 0 $ that
    $$
    \calN(S^{n-1}, \epsilon) \le \left( 1 + \frac{2}{\epsilon} \right)^n .
    $$
\end{lem}

Then we have $|\calN| \le 9^{d_1d_2}$. Considering the union bound:
\begin{align*}
\bbP \left( \max_{v\in \calN} \left| \frac{1}{nK} \|Av\|_2^2 - 1 \right| \le \frac{\epsilon}{2} \right) &\le 9^{d_1d_2} \cdot 2 \exp \left[ - \frac{c}{64\sigma^2 K^2 S_v^4} (C^2 d_1d_2 + t^2) \right] \\
&\le 2 \exp (-c_1 t^2) ,
\end{align*}
for $C\ge 8\sigma K S_v^2 \sqrt{\log{9}/c}$ and $c_1 = c/(64\sigma^2 K^2S_v^4)$. So we have with probability $1 - 2\exp(-c_1 t^2)$, 
$$
\left\|\frac{1}{nK} A^\top A - I \right\| \le \max(\delta, \delta^2) .
$$

Now note that for a non-negative matrix $M$, if $\|M- I\|_2 \le \epsilon$, then for arbitrary $v$ such that $\|v\|= 1$,
$$
    \|Mv\| \ge \|v\| - \|(M-I)v\| \ge 1- \epsilon .
$$

Therefore $\lambda_{\min}(M) \ge 1- \epsilon$. Now assume $\delta \le 1$, then we have with probability $1- 2\exp(-c_1t^2)$,
$$
\lambda_{\min}(A^\top A) \ge nK - nK \left( C \sqrt{\frac{d_1d_2}{n}} + \frac{t}{\sqrt{n}} \right) = nK - K\sqrt{n} (C\sqrt{d_1d_2} + t) .
$$

Note that $A^\top A = \Sigma^{-1/2}V_n \Sigma^{-1/2}$. Therefore, 
$$
\lambda_{\min}(V_n) \ge \lambda_{\min}(\Sigma) nK - \lambda_{\min}(\Sigma)K \sqrt{n} (C\sqrt{d_1d_2} + t) ,
$$
and similarly:
$$
\|V_n\| \ge \|\Sigma\|_2 (nK - K\sqrt{n}(C\sqrt{d_1d_2}+ t)) .
$$

Note that if $n \ge (2b/a)^2 + 2B/a$,
\begin{align*}
    an - b\sqrt{n} & \ge \frac{an}{2} + a\left(\frac{n}{2} - \frac{b}{a}\sqrt{n}\right) \\
    &\ge \frac{a}{2} \cdot n + \frac{a\sqrt{n}}{2}\left(\sqrt{n} - \frac{2b}{a}\right) \\
    &\ge B.
\end{align*}

Therefore, if 
$$
n \ge 4\sigma^2 (C_1\sqrt{d_1d_2} + C_2\sqrt{\log(1/\delta)})^2 + \frac{2B}{\lambda_{\min}(\Sigma)K} ,
$$
for constants $C_1, C_2$  defined as $C_1 = 4 K S_v^2 \sqrt{\log{9}/c} $ and $C_2 = {8KS_v^2}/{\sqrt{c}}$, then with probability at least $1-\delta$ we have $\lambda_{\min}(V_n) \ge B$.

Also note that 
\begin{align*}
    \frac{1}{nK} \|V_n\|_2 &= \|\Sigma^{1/2} \frac{1}{nK} A^\top A \Sigma^{1/2}\|_2 \\
    &\ge \|\Sigma^{1/2} I \Sigma^{1/2}\|_2 - \|\Sigma^{1/2} \left(\frac{1}{nK}A^\top A - I\right)\Sigma^{1/2} \|_2 \\
    &\ge \|\Sigma\|_2 - \|\Sigma^{1/2}\|_2 \cdot \left\|\frac{1}{nK}A^\top A - I\right\|_2 \|\Sigma^{1/2}\|_2 \\
    &\ge \|\Sigma\|_2 ( 1 - \epsilon) .
\end{align*}

Therefore,
$$
\frac{\|V_n\|_2}{n} \ge K \|\Sigma\|_2 \left(1 - \frac{1}{\sqrt{n}}(C_1\sqrt{d_1d_2} + C_2\sqrt{\log(1/\delta)}) \right) .
$$

Assuming $n \ge 4(C_1\sqrt{d_1d_2} + C_2 \log(1/\delta))^2$, we have
$$
\frac{\|V_n\|_2}{n} \ge \frac{1}{2} K\|\Sigma\|_2 .
$$

\subsection{Proof of Proposition \ref{prop:lowranksubspace}}
\label{sect:A2}

Proposition \ref{prop:lowranksubspace} provides bounds on the error induced by truncation. We begin with the following Lemma from \citet{jun2019bilinear}.

\begin{lem}[Wedin's sin$\Theta$ theorem, Appendix B of \citep{jun2019bilinear}] \label{lem:wedin}
    Let $\hat{\Theta} = [\hat{U}, \hat{U}_\perp] \hat{D} [\hat{V}, \hat{V}_\perp]^\top$ and $\Theta^* = U^* D^* V^*$ be the SVD of $\hat\Theta$ and $\Theta^*$. Then,
    $$
    \|\hat{U}_\perp^\top U^* \|_F \|\hat{V}_\perp^\top V^*\|_F \le \frac{\|\hat{\Theta} - \Theta^*\|_F^2}{s_r^{*2}} ,
    $$
    where $s_k^*$ is the $k$-th (smallest) eigenvalue of $\Theta^*$. 
\end{lem}

Now let $\theta^*$ be the rotation-vectorization of $[\hat{U}, \hat{U}_\perp]^\top \Phi^* [\hat{V}, \hat{V}_\perp]$. Then we have
$$
\|\theta^*_{k+1:p}\|_2 \le \|\hat{U}_\perp U\|_F \|\hat{V}_\perp V\|_F \cdot \|D^*\|_2 \le \frac{s_1^*}{s_r^{*2}} \|\hat{\phi} - \phi^*\|^2 .
$$

Combining Proposition 1 and Lemma \ref{lem:wedin}, the proof is complete.

\subsection{Proof of Theorem 1}
\label{sect:A3}

Theorem 1 provides bounds on the cumulative regret of our proposed algorithm. In this section we provide a more detailed proof of the theorem. First we restate the following result from \citet{oh2021multinomial}.

\begin{lem} [Lemma 6 from \citet{oh2021multinomial}]
    If $\lambda_{min}(W_{T_0}) \ge K$, then
    $$
    \sum_{t=1}^T \max_{i\in S_t} \|x_{it, rtv}\|_{W_t^{-1}}^2 \le 2k \log(T/k) .
    $$
\end{lem}

Now define the event $\calE_t := \{\|\hat\theta_{n,rtv} - \theta^*_{rtv}\| \le \alpha_n\}$. Then conditioned on $\calE_t$, we have:
\begin{align*}
    r_t &= R(\{\langle X_{it}, \Theta^* \rangle\}, S^*) - R(\{\langle X_{it}, \Theta^*\rangle\}, S_t) \\
    &\le R(\{z_{it}\}, S^*) - R(\{\langle X_{it}, \Theta^*\rangle\}, S_t) \\
    &\le R(\{z_{it}\}, S_t) - R(\{\langle X_{it}, \Theta^* \rangle\}, S_t) \\
    &\le R_{\max} \cdot \max_{i\in S_t} (z_{it} - \langle X_{it}, \Theta^* \rangle) \\
    &\le R_{\max} \cdot 2\alpha_t \max_{i\in S_t}\|x_{it,rtv}\|_{W_t^{-1}} + 2S_2^2 S_\perp ,
\end{align*}
where the last line holds from Lemma \ref{lem:zit_bound}. Finally, let $\calE$ be the event where $\|\Theta^*_{r+1:d_2, r+1:d_1}\|_F^2 \le S_\perp$ holds, using Proposition \ref{prop:lowranksubspace} with $\delta = T^{-1}$, we know that $\bbP(\calE^c) \le T^{-1}$. Moreover we know that $\bbP(\calE_t^c) \le t^{-2}$ from Lemma 1. Therefore, we can bound the overall regret:
\begin{align*}
    {\calR}_T &\le R_{\max} T \cdot \bbP({\calE}^c) + \bbE\left[ \sum_{t=1}^T r_t \Big| \calE \right] \cdot \bbP(\calE) \\
    &\le R_{\max} +  R_{\max} T_0 + R_{\max} \cdot \sum_{t=T_0+1}^T \left( 2\alpha_t \max_{i\in S_t} \|x_{it,rtv}\|_{W_t^{-1}} + 2S_2^2 \cdot S_\perp \right) + \sum_{t=T_0+1}^T R_{\max}\cdot \bbP({\calE}_t^c) \\
    &\le R_{\max} \left( T_0 + 1+ 2 \alpha_T \sum_{t=T_0 + 1}^T \max_{i\in S_t} \|x_{it,rtv}\|_{W_t^{-1}} + 2S_2^2 S_\perp (T-T_0) \right) + R_{\max} \sum_{t=1}^T t^{-2}\\
    &\le R_{\max} \left( T_0 + 3 + 2 \alpha_T \sqrt{(T-T_0) \sum_{t=T_0+1}^T \|x_{it,rtv}\|_{W_t^{-1}}^2} + 2S_2^2 S_\perp(T-T_0
    ) \right) \\
    &\le R_{\max} \left( T_0 + 3 + 2 \alpha_T \sqrt{T \cdot 2k \log(T/k)} + 2S_2^2 S_\perp(T-T_0) \right) ,
\end{align*}
which completes the proof. In the remaining section, we provide proof of Lemma 1 and \ref{lem:zit_bound}.

\subsubsection{Proof of Lemma 1}

Define 
$$
J_n(\theta_{rtv}) = \sum_{t=1}^n \sum_{i\in S_t} \left( p_t(i| x_{jt,rtv}^\top \theta_{rtv}) - p_t(i|x_{jt}^\top \theta^*) \right) x_{it,rtv} .
$$

Then, if $\hat\theta_{n,rtv}$ is the MLE of the truncated problem,
$$
Z_n = J_n(\hat\theta_{n,rtv}) = \sum_{t=1}^n \sum_{i\in S_t} \epsilon_{it} x_{it,rtv} ,
$$
where $\epsilon_{it} = y_{it} - p_t(i|x_{jt}^\top \theta^*)$. Note that $\epsilon_{it}$ is $\sigma$-subgaussian with $\sigma = 1/2$. Now let $W_n = \frac{1}{nK}\sum_{t=1}^n \sum_{i\in S_t} x_{it,rtv} x_{it,rtv}^\top$. Following the proof of Lemma 2 in \citep{oh2021multinomial}, we have :
$$
\|Z_n\|_{W_n^{-1}}^2 \le \frac{2\sigma^2}{nK} \left[ k \log\left(1+ \frac{n}{k}\right) + \log\frac{1}{\delta}\right] ,
$$
with probability at least $1-\delta$. Also following the argument on the consistency of MLE on \citep{oh2021multinomial}, we have
$$
\|J_n(\hat\theta_{n,rtv}) - J_n(\theta^*_{rtv})\|_{W_n^{-1}}^2 \ge \frac{\kappa^2}{n^2K^2} \|\hat\theta_{n,rtv} - \theta^*_{rtv}\|_{W_n} .
$$

However, since this is the truncated space, we do not have $J_n(\theta^*_{rtv}) = 0$. Instead,
\begin{align*}
    J_n(\theta^*_{rtv}) &= \sum_{t=1}^n \sum_{i\in S_t} \left( p_t(i|x_{jt,rtv}^\top \theta^*_{rtv}) - p_t(i|x_{jt}^\top \theta^*) \right) x_{it,rtv} .
\end{align*}

Now note that $p_t(i|u_j)$ is actually the expected reward at time $t$ given the reward $r_j = 1_{\{i=j\}}$. Using Lemma \ref{lem:optimistic_assortment}, we have
$$
|p_t(i|x_{jt,rtv}^\top \theta^*_{rtv}) - p_t(i|x_{jt}^\top \theta^*)| \le \max_{i\in S_t} |x_{it,rtv}^\top \theta^*_{rtv} - x_{it}^\top \theta^*| \le S_2^2 \cdot S_\perp .
$$

Therefore,
\begin{align*}
    \|J_n(\theta^*_{rtv})\|_{W_n^{-1}}^2 &\le S_2^4 \cdot S_\perp^2 \left(\sum_{t=1}^n \sum_{i\in S_t} |x_{it,rtv}| \right)^\top W_n^{-1} \left(\sum_{t=1}^n \sum_{i\in S_t} |x_{it,rtv}| \right) \\
    &\le S_2^4 \cdot S_\perp^2 \cdot\frac{1}{nK} .
\end{align*}

Then we have:
\begin{align*}
    \|\hat\theta_{n,rtv} - \theta^*_{rtv}\|_{W_n} &\le \frac{1}{\kappa} \left( \|Z_n\|_{W_n^{-1}} + \|J_n(\theta^*_{rtv})\|_{W_n^{-1}} \right) \\
    &\le \frac{1}{\kappa nK} \left( \sigma \sqrt{2k \log\left( 1 + \frac{n}{k} \right) + 4\log{n}} + S_2^2 \cdot S_\perp \cdot \sqrt{nK} \right) =: \alpha_n ,
\end{align*}
by plugging in $\delta = 1/n^2$.

\subsubsection{Proof of Lemma \ref{lem:zit_bound}}

Adapting Appendix F.1 of \citep{kang2022efficient}, let $S_t = \argmax_{S'} R(\{z_{it}\}, S')$, and assume $\|\widehat{\theta}_{t,rtv} - \theta^*_{rtv}\|_{W_n} \le \alpha_t$. let $z_{it}  = x_{it,rtv}^\top \widehat{\theta}_{t,rtv} + \beta_{it}$ where $\beta_{it} = \alpha_t \cdot \|x_{it,rtv}\|_{W_t^{-1}} + S_2^2\cdot S_\perp$. Then,
$$
z_{it} - \langle X_{it}, \Theta^* \rangle = x_{it,rtv}^\top \hat\theta_{rtv} - x_{it,rtv}^\top \theta^*_{rtv} + x_{it,rtv}^\top \theta^*_{rtv} - \langle X_{it}, \Theta^* \rangle + \beta_{it} ,
$$
where
$$
|x_{it,rtv}^\top \hat\theta_{rtv} - x_{it,rtv}^\top \theta^*_{rtv}| \le \|x_{it,rtv}\|_{W_t^{-1}} \cdot \|\hat\theta_{rtv} - \theta^*_{rtv}\|_{W_t} \le \alpha_t \|x_{it,rtv}\|_{W_t^{-1}} ,
$$
and
$$
|x_{i,rtv}^\top \theta^*_{rtv} - \langle X_{it}, \Theta^* \rangle| \le \|\hat{U}_\perp^\top X_{it} \hat{V}\|_F \cdot \|\hat{U}_\perp^\top \Theta^* \hat{V}\|_F \le S_2^2 \cdot S_\perp .
$$

Therefore,
$$
0 \le z_{it} - \langle X_{it}, \Theta^* \rangle \le 2 \beta_t = 2  \alpha_t \|x_{it,rtv}\|_{W_t^{-1}} + 2S_2^2 \cdot S_\perp,
$$
with probability at least $1-t^{-2}$. 

\subsection{Proof of Corollary \ref{cor:regretbound_hp}}\label{sect:A4}

Here we prove the high probability regret bound Corollary \ref{cor:regretbound_hp}. The proof directly follows from the proof of Theorem 1.
\begin{proof}
    Note that $\bbP(\calE_t^c) \le t^{-2}$. Hence, $\bbP(\cup_{t=T_0+1}^T \calE_t^c) \le \sum_{t=T_0+1}^T t^{-2} \le \int_{T_0}^\infty t^{-2} dt = T_0^{-1}$. Therefore, $\bbP(\cap_{t=T_0+1}^T \calE_t) \ge 1 -T_0^{-1}$. Hence,
    $$
    r_t \le R_{\max} \cdot 2\alpha_t \max_{i\in S_t} \|x_{it,rtv}\|_{W_t^{-1}} + 2S_2^2 S_\perp,
    $$
    for every $t >T_0$ with probability at least $1-T_0^{-1}$. Hence we have the high probability bound:
    \begin{align}
        \sum_{t=1}^T r_t &\le R_{\max} \left( T_0 + 2\alpha_T \sqrt{T \cdot 2k\log(T/k)} + 2S_2^2 S_\perp (T-T_0)\right),
    \end{align}
    with probability at least $1-\delta - T_0^{-1}$. Assuming $T_0 \ge 1/\delta$, the proof is complete.
\end{proof}

\subsection{Proof of Proposition \ref{prop:rankest_gic}}\label{sect:A5}

In this section we prove Proposition \ref{prop:rankest_gic}, the rank estimation consistency of GIC. As GIC is a sum of the loss function (negative log likelihood) and the penalty, we prove that 1) when the rank is under-estimated, the loss is larger at a constant rate compared to the true rank; 2) when the rank is over-estimated the loss decrease is much smaller than the penalty increase. Combining these two results, the GIC is minimized at the true rank, when sample size is sufficiently large.

\begin{proof}

Let $\calL_n$ be the negative log likelihood given as
$$
{\calL}_n(\Phi) = \frac{1}{n} \sum_{t=1}^n \left[\sum_{i\in S_t} y_{it} p_i^\top \Phi q_t - \log \left( 1 + \sum_{j\in S_t} \exp(p_j^\top \Phi q_t) \right) \right],
$$
and let $\Phi_r$ be the rank-$r$ constrained MLE given as
$$
\hat\Phi_r := \argmin_{\rk(\Phi) \le r} {\calL}_n (\Phi).
$$

Denote $d = \min\{d_1, d_2 \}$. Note that $\hat\Phi = \hat\Phi_d = \argmin \calL_n(\Phi)$ is the unconstrained MLE. Therefore $\nabla \calL_n (\hat\Phi) = 0$. Now using Taylor expansion we have
\begin{align*}
    {\calL}_n (\Phi) - {\calL}_n (\hat\Phi) &= (\phi - \hat\phi)^\top \nabla^2 {\calL}_n(\tilde\phi) (\phi - \hat\phi).
\end{align*}

And also note that 
$$
\nabla^2 {\calL}_n (\Phi) = \frac{1}{n} \sum_{t=1}^n \left( \sum_{i\in S_t} p_{it} z_{it} z_{it}^\top - \sum_{i,j\in S_t} p_{it} p_{jt} z_{it} z_{jt}^\top \right).
$$

\begin{lem}[Strong Convexity and Smoothness]
    $$
    \frac{\kappa}{n} \lambda_{\min}(V_n) \|\phi - \hat\phi\|^2 \le {\calL}_n (\Phi) - {\calL}_n (\hat\Phi) \le \frac{1}{n} \lambda_{\max}(V_n) \|\phi - \hat\phi\|^2 ,
    $$
    or with probability at least $1-\delta$,
    $$
    C_1 \|\phi - \hat\phi\|^2 \le  {\calL}_n (\Phi) - {\calL}_n (\hat\Phi) \le C_2 \|\phi - \hat\phi\|^2.
    $$
\end{lem}

The strong convexity and smoothness result follows from Lemma \ref{lem:RSC} and \ref{lem:lamVn}. 

\begin{lem}[Estimation error for overestimated rank]
    For $k \ge r^*$, we have
    $$
    \epsilon_k^2 := \|\hat\Phi_k - \Phi^*\|_F^2 \le C \cdot \frac{d_1 d_2}{n}.
    $$
\end{lem}

\begin{proof}
    Note that for $k \ge r^*$,
    $$
    {\calL}_n(\hat\Phi) \le {\calL}_n (\hat\Phi_k) \le {\calL}_n (\hat\Phi_r).
    $$

    Therefore,
    \begin{align*}
        C_1 \|\hat\Phi_k - \hat\Phi\|_F^2 &\le {\calL}_n (\hat\Phi_k) - {\calL}_n (\hat\Phi) \\
        &\le {\calL}_n(\hat\Phi_{r^*}) - {\calL}_n (\hat\Phi) \\
        &\le C_2 \|\hat\Phi_{r^*} - \hat\Phi\|_F^2 \\
        &\le O\left( \frac{(d_1 + d_2)r}{n} \right),
    \end{align*}

    and hence, $\|\hat\Phi_k - \hat\Phi\|_F = O\left( \sqrt{\frac{(d_1 + d_2)r}{n}}\right)$. Note that $\|\hat\Phi - \Phi^*\|_F^2 = O(\frac{d_1d_2}{n})$. Therefore, $\|\hat\Phi_k - \Phi^*\|_F = O\left(\sqrt{\frac{d_1d_2}{n}}\right)$.
\end{proof}

\begin{lem}[Estimation error for underestimated rank]
    For $k \le r^*$, we have
    $$
    \epsilon_k^2 := \|\hat\Phi_k - \Phi^*\|_F^2 \ge C \cdot \sum_{i=k+1}^{r^*} \sigma_i^2.
    $$
\end{lem}

Then for $r < r^*$, we have:
\begin{align*}
    {\calL}_n (\hat\Phi_r) - {\calL}_n (\hat\Phi_{r^*}) &= ({\calL}_n (\hat\Phi_r) - {\calL}_n (\hat\Phi)) - ({\calL}_n (\hat\Phi_{r^*}) - {\calL}_n(\hat \Phi)) \\
    &\ge C_1 \|\hat\phi_r - \hat\phi\|^2 \\
    &\ge C_1 (\epsilon_r - \epsilon_d)^2 \\
    &\to C_1 \epsilon_r^2.
\end{align*}

Similarly, for $r > r^*$, we have:
\begin{align*}
    {\calL}_n (\hat\Phi_r) - {\calL}_n (\hat\Phi_{r^*}) &= ({\calL}_n (\hat\Phi_r) - {\calL}_n (\hat\Phi)) - ({\calL}_n (\hat\Phi_{r^*}) - {\calL}_n(\hat \Phi)) \\
    &\ge C_1 \|\hat\phi_r - \hat\phi\|^2 - C_2 \|\hat\phi_{r^*} - \hat\phi\|^2\\
    &\ge C_1 (\epsilon_r - \epsilon_d)^2 - C_2 (\epsilon_{r^*} + \epsilon_d)^2.
\end{align*}

and therefore, ${\calL}_n (\hat\Phi_{r^*}) - {\calL}_n (\hat\Phi_{r}) = O(\epsilon_d^2) = O(\frac{d_1d_2}{n})$.

Now define the penalty $P(r)$ as
$$
P(r) = a_n \cdot (d_1 + d_2 - r) \cdot r.
$$

Then,
$$
P(r) - P(r^*) = a_n (d_1 + d_2 - r - r^*) (r-r^*).
$$

Note that GIC will be defined as
$$
GIC(r) := {\calL}_n (\hat\Phi_r) + P(r).
$$

If we define $a_n = \log(n)/n$, $GIC(r) > GIC(r^*)$ as $n\to \infty$ for $r \ne r^*$.
\end{proof}

\bibliography{ref}

\end{document}